\documentclass[12pt,reqno]{amsart}
\usepackage{fullpage}
\usepackage{amsmath,amssymb,amsthm,mathtools}
\usepackage{colonequals,booktabs}
\usepackage[colorlinks=true,linkcolor=blue,citecolor=blue,urlcolor=blue]{hyperref}

\newtheorem{theorem}{Theorem}[section]
\newtheorem{lemma}[theorem]{Lemma}
\newtheorem{proposition}[theorem]{Proposition}
\newtheorem{corollary}[theorem]{Corollary}
\theoremstyle{remark}
\newtheorem{remark}[theorem]{Remark}
\newcommand{\R}{\mathbb R}
\newcommand{\Z}{\mathbb Z}
\newcommand{\E}{\mathbb E}
\newcommand{\C}{\mathbb C}
\newcommand{\dd}{\,\mathrm d}
\newcommand{\1}{\mathbf 1}
\newcommand{\ip}[2]{\langle #1,#2\rangle}
\newcommand{\norm}[1]{\lVert #1\rVert}
\renewcommand{\subset}{\subseteq}
\newcommand{\embolden}[1]{\textbf{#1}}

\newcommand{\B}{\mathbb B}
\newcommand{\Stab}{\operatorname{Stab}}

\renewcommand{\epsilon}{\varepsilon}

\title{Sharp Hardness for MAX-3-CUT and Quantum MAX-CUT}
\author{Steven Heilman}
\date{\today}

\thanks{
Email: stevenmheilman@gmail.com\\
Supported by NSF Grant CCF AF 2448108.  \\
2020 Mathematics Subject Classification: 68W27, 65K10, 68Q32\\
Keywords: MAX-CUT, MAX-3-CUT, hardness of approximation, Unique Games Conjecture\\
Department of Mathematics, University of Southern California, Los Angeles, CA 90089}
\begin{document}

\begin{abstract}
Assuming the Unique Games Conjecture, we show it is NP-hard to approximate MAX-3-CUT within  a multiplicative factor of $\alpha_3+\epsilon$ for every $\epsilon>0$, where $\alpha_3\approx.83600811464$ is the approximation ratio of Frieze-Jerrum's polynomial-time algorithm from 1995.  That is, we prove sharp hardness of approximation for MAX-3-CUT.  This result resolves a conjecture of Khot-Kindler-Mossel-O'Donnell from 2004 by proving the three candidate Plurality is Stablest Conjecture for correlations in $[-1/2,2/5]$ and generalizes the Majority is Stablest Theorem of Mossel-O'Donnell-Oleszkiewicz [Annals of Math, 2010].

With a similar strategy we prove: assuming the Unique Games Conjecture, it is NP-hard to approximate the product-state value of Quantum MAX-CUT within a multiplicative factor of $\alpha_{\rm BOV}+\epsilon$ for every $\epsilon>0$, where $\alpha_{\rm BOV}\approx 0.9563372685$ is the approximation ratio of the Bri\"et-de Oliveira Filho-Vallentin algorithm.  This sharp hardness result completes the conjectured hardness of Hwang-Neeman-Parekh-Thompson-Wright from 2021 by proving their $S^{k-1}$-valued Borell inequality for correlations in $[-.5843,.5843]$ for all $k\geq3$.
\end{abstract}

\maketitle

\section{Introduction}

The MAX-CUT problem is arguably one of the simplest and most thoroughly studied combinatorial optimization problems.  It is one of Karp's original 21 NP-complete problems \cite{karp72}.  In this paper we make progress on understanding the computational hardness of two of its natural generalizations: MAX-3-CUT and Quantum MAX-CUT.  For these problems, efficient algorithms are known, but understanding their optimality was a significant challenge.

The MAX-CUT problem asks for the partition of the vertices of a finite undirected graph into two sets that maximizes the number of edges going between the two sets.  The result of \cite{cres01} shows that the weighted and unweighted versions of MAX-CUT have the same hardness of approximation (and likewise for MAX-3-CUT).  The MAX-CUT problem is NP-hard, so if P$\neq$NP, no polynomial-time algorithm can solve it.  It is easy to find a partition of the vertices into two sets that is at least $50\%$ as good as the best possible, by random selection of vertices into one of the sets.  This algorithm can be derandomized \cite{sahni76}.  However, achieving a quality of approximation beyond 50\% appeared to be a difficult problem.

A highly nontrivial improvement appeared in 1995: for any $\epsilon>0$, the semidefinite programming algorithm of \cite{goemans95} approximates the MAX-CUT problem in polynomial-time within a multiplicative factor of $\alpha_{2}-\epsilon$, where
\begin{equation}\label{a2def}
\alpha_{2}
\colonequals
\inf_{-1\leq\rho\leq 1}\frac{2}{1}\frac{\frac{1}{\pi}\arccos(\rho)}{1-\rho}\\
\approx.87856720578.
\end{equation}
That is, it is possible to obtain in polynomial-time a partition of the graph where the number of edges going between the two partition elements is at least 87.8\% as much as the largest partition value.

The approximation factor \eqref{a2def} remains the best possible in polynomial-time.  It is then desirable to obtain a matching hardness result.  MAX-CUT was proven MAX-SNP hard in 1991 \cite{papa91}, and in 1998, it was shown \cite{bellare98} that approximating MAX-CUT within a factor of $83/84$ is NP-hard.  This was improved to $16/17\approx.9412$ in \cite{hastad01}.  Despite not reaching the $.878$ value, this is the best known unconditional NP-hardness result for MAX-CUT.  It was also shown in \cite{feige02} that integrality gap instances exist for the Goemans-Williamson algorithm, i.e. graphs exist that achieve the worst case behavior of this algorithm.

In a breakthrough work, \cite{khot07} showed that, assuming the Unique Games Conjecture, the gap between the $16/17$ hardness and the $.87856\ldots$ approximation can be closed exactly.

\begin{theorem}[Sharp Hardness for MAX-CUT, {\cite[Theorem 1]{khot07}}]\label{thm8}
Assume that the Unique Games Conjecture is true \cite{khot02,khot10a,khot18}.  Then, for any $\epsilon>0$, it is NP-hard to approximate MAX-CUT within a multiplicative factor of $\alpha_{2}+\epsilon$.
\end{theorem}

Theorem \ref{thm8} is called a sharp hardness result since the quantity $\alpha_{2}$ gives an exact barrier between tractability (i.e. the polynomial-time algorithm of \cite{goemans95}) and intractability (as stated in Theorem \ref{thm8}).  Theorem \ref{thm8} relied on the Majority is Stablest Theorem and Invariance Principle of \cite{mossel10}.  The invariance principle showed that the discrete Majority is Stablest Theorem is equivalent to the continuous Borell's inequality \cite{borell85}, a known fact from Gaussian geometry proven 20 years prior.  In other words, \cite{mossel10} showed that the discrete problem they needed was equivalent to a known continuous problem.  Other examples of invariance principles used in different contexts include \cite{chatterjee06} and \cite{rotar79}.

For a statement of the Unique Games Conjecture, see \cite{khot02,khot10a}.  For some recent progress demonstrating that the conjecture is ``halfway proven,'' see \cite{khot18}.

Before the Unique Games Conjecture was introduced, the PCP Theorem \cite{arora98,arora98b,hastad01} implied that any constraint satisfaction problem (CSP) should be NP-hard to approximate with some multiplicative approximation less than $1$.  This led to some optimal inapproximability results \cite{hastad01}.  For example, MAX-3-SAT has an elementary efficient $7/8$ approximation algorithm that chooses variable assignments randomly, but for any $\epsilon>0$, a $7/8+\epsilon$ approximation is NP-hard \cite{hastad01}.  Similarly, MAX-3-LIN(2) has an elementary efficient $1/2$ approximation algorithm that chooses variable assignments randomly, but for any $\epsilon>0$, a $1/2+\epsilon$ approximation is NP-hard \cite{hastad01}.  The problems MAX-3-SAT and MAX-3-LIN(2) are therefore approximation resistant: it is hard to approximate them better than simple random guessing.  Other sharp NP-hardness of approximation results include: Maximum k-coverage \cite{feige98}, Metric k-center \cite{hoch85} and set cover \cite{dinur14}, the latter relying on parallel repetition.

Some optimal inapproximability results that rely on the Unique Games Conjecture (besides Max-CUT) include: MAX-2-SAT \cite{brak24}, Min Vertex Cover \cite{khot08}, Grothendieck's inequality \cite{ragh09}, kernel clustering \cite{khot10}, etc.  In these problems, unlike MAX-3-SAT and MAX-3-LIN(2), random guessing does not lead to an optimal algorithm.  Raghavendra's result \cite{ragh08} shows that essentially any CSP will have some sharp inapproximability, assuming the Unique Games Conjecture holds.  However, this result does not determine what the constant of inapproximability will be for any particular problem.

\subsection{MAX-3-CUT}

The MAX-3-CUT problem asks for the partition of the vertices of a finite undirected graph into three sets that maximizes the number of edges going between the three sets.  This is a natural generalization of MAX-CUT to three sets of vertices.  This problem is equivalent to the Maximum 3-Colorable Subgraph problem.  As we will see, the analysis of MAX-3-CUT reveals surprising difficulties that were not apparent for MAX-CUT itself.  In particular, efficient approximation algorithms for this problem were developed in succession with those for MAX-CUT, but the sharp hardness has remained an open problem until now.

A random assignment of vertices to parts of the partition yields a $2/3$ multiplicative approximation to MAX-3-CUT, so an approximation larger than $2/3$ yields a nontrivial algorithm.

In 1997, Frieze and Jerrum \cite{FriezeJerrum1997} found an $\alpha_3$ approximation algorithm for MAX-3-CUT by generalizing the semidefinite programming algorithm of \cite{goemans95}, where
\begin{equation}\label{alpha3def}
\begin{aligned}
\alpha_{3}
&\colonequals\inf_{-\frac{1}{2}\leq\rho\leq 1}\frac{3}{2}\frac{1-3\Big(\frac{1}{9}+\frac{[\arccos(-\rho)]^{2} - [\arccos(\rho/2)]^{2}}{4\pi^{2}}\Big)}{1-\rho}\\
&=\frac{3}{2}\frac{1-3\Big(\frac{1}{9}+\frac{[\arccos(1/2)]^{2} - [\arccos(-1/4)]^{2}}{4\pi^{2}}\Big)}{1+1/2}
\approx.83600811464.
\end{aligned}
\end{equation}

As with the Goemans-Williamson algorithm for MAX-CUT, this remains the best known polynomial-time approximation algorithm for MAX-3-CUT.  And it again becomes desirable to obtain a matching hardness result.

In 1996 \cite{kann97} proved that approximating MAX-3-CUT within a multiplicative factor of $101/102$ is NP-hard.  This was improved to $67/68$ in \cite{trevisan00}, and further improved to $32/33$ in \cite{gurus13}.  The best known NP-hardness of approximation is now $16/17$ from 2014 \cite{austrin14}.  Also \cite{nagda25} obtained a gadget-based reduction with a worse approximation factor of $55/57$.  As with MAX-CUT, there is a gap between the best known algorithm's performance and the best known NP-hardness result.

In 2009, \cite{isaksson12} attempted to generalize the proof of sharp Unique Games hardness \cite{khot07} from MAX-CUT to MAX-3-CUT.  They succeeded in applying the invariance principle machinery from \cite{mossel10} to MAX-3-CUT.  Recall that, in the case of MAX-CUT, the required Majority is Stablest Theorem was shown to be equivalent to the known Borell's inequality.  However, MAX-3-CUT hardness requires the Plurality is Stablest Conjecture, which is equivalent to a continuous but unproven inequality known as the Standard Simplex Conjecture \cite{isaksson12}.

Our first main result, Theorem \ref{thm1} below, proves the Standard Simplex Conjecture of \cite{isaksson12} for three sets for correlations in $[-1/2,2/5]$.  We state an informal version here.

\begin{theorem}[Standard Simplex Conjecture, Informal, Three Sets, $-1/2\leq\rho\leq2/5$]\label{thm1}
Let $-1/2\leq\rho<0$.  The three sets in Euclidean space of dimension $n\geq2$ minimizing noise stability with correlation $\rho$ are $(\Theta_i\times\R^{n-2})_{i=1}^{3}$ where $(\Theta_i)_{i=1}^{3}\subset\R^2$ are three disjoint $120$ degree sectors centered at the origin.  Up to rotations and measure zero changes, these are the only minimizing sets.

If $0<\rho\leq2/5$, these three sets maximize noise stability among partitions each of whose sets have Gaussian measure $1/3$.
\end{theorem}

The case $\rho=0$ is uninteresting, since the noise stability of a set at correlation zero is just the squared Gaussian measure of the set.

Theorem  \ref{thm1} is proven by combining Theorems \ref{thm:main} and \ref{thm:positive-companion} below.

Due to the equivalence results of \cite[Theorem 1.10]{isaksson12}, a standard simplex conjecture proof implies also a Plurality is Stablest result.  For now, we keep the statement somewhat informal.  For a more formal statement, see \cite{isaksson12,Heilman2023}.

\begin{theorem}[Plurality is Stablest, Informal, Three Candidates, $-1/2\leq\rho\leq2/5$]\label{thm2}
Let $-1/2\leq\rho<0$.  The plurality function has the smallest noise stability with correlation $\rho$ among all low-influence functions $f\colon\{1,2,3\}^n\to\{1,2,3\}$.

Let $0<\rho\leq2/5$.  The plurality function has the largest noise stability with correlation $\rho$ among all low-influence functions with $\mathbb{P}(f=i)=1/3$ for all $i=1,2,3$.
\end{theorem}

As shown in \cite{isaksson12}, Theorem \ref{thm1} or \ref{thm2} with $\rho=-1/2$ implies our first main application.

\begin{theorem}[\embolden{Sharp Hardness for MAX-3-CUT}]\label{thm3}
Assume the Unique Games Conjecture holds.  Then, for any $\epsilon>0$, it is NP-hard to approximate MAX-$3$-CUT within
$\alpha_3+\varepsilon$.
\end{theorem}

Theorem \ref{thm3} therefore completes the strategy initiated in 2004 in \cite{khot07} for sharp MAX-3-CUT hardness.

Theorem \ref{thm2} can also be applied to the MAX-2-LIN(3) problem, which asks to maximize the number of satisfied linear equations among two-term linear equations in the field of three elements.  The positive correlation case of Theorem \ref{thm2} implies a hardness result for MAX-2-LIN(3), due to \cite[Theorem 12]{khot07}.   However, the optimal correlation choice of around $\rho\approx.615$ is outside the range of Theorem \ref{thm2}, so the hardness result is not optimal.  Nevertheless, Theorem \ref{thm:max-2lin-3-hardness} appears to be the best known hardness result for this problem.

\begin{theorem}[Hardness for $\mathrm{MAX\text{-}2\text{-}LIN}(3)$]
\label{thm:max-2lin-3-hardness}
Assume the Unique Games Conjecture.  Then, for every $\varepsilon>0$,
it is NP-hard to distinguish instances $\mathcal I$ of
$\mathrm{MAX\text{-}2\text{-}LIN}(3)$ satisfying
$
 \operatorname{OPT}(\mathcal I)\geq \frac35-\varepsilon
$
from those satisfying
$
 \operatorname{OPT}(\mathcal I)
 \leq
 \frac13+\frac{3}{4\pi^2}
 \left(
   \arccos(-2/5)^2-\arccos(1/5)^2
 \right)+\varepsilon.
$
Consequently, assuming the Unique Games Conjecture, it is NP-hard to
approximate $\mathrm{MAX\text{-}2\text{-}LIN}(3)$ within a
multiplicative factor of
\[
 \alpha_{2\rm LIN(3)}+\epsilon.
\]
where $\alpha_{2\rm LIN(3)}\approx0.815723527338279$.  The result holds even when every equation has the form
$x_i-x_j=c$ $\rm mod$ $3$.
\end{theorem}

Here the constant appearing in the hardness result is
\[
 \alpha_{\mathrm{2LIN}(3)}
 \colonequals
 \frac59+\frac{5}{4\pi^2}
 \left(
   \arccos(-2/5)^2-\arccos(1/5)^2
 \right).
\]

Note that MAX-3-CUT can also be understood by analogy with MAX-2-LIN(3) as maximizing the number of satisfied expressions of the form $x_i-x_j\neq0$ mod 3.

Hardness for MAX-3-CUT and MAX-2-LIN(3) is partly motivated by their relation to the Unique Games Conjecture (UGC) itself.  Suppose we are given $m$ two-term linear equations in indeterminates $x_1,\ldots,x_n$ of the form $x_i - x_j=c_{ij}$ mod $p$ for some prime $p>0$.  Then UGC says: for any $\epsilon>0$, there exists a prime $p=p(\epsilon)$ such that it is NP-hard to decide if at least $(1-\epsilon)m$ equations are satisfiable, or at most $\epsilon m$ equations are satisfiable.  (If neither situation occurs, an arbitrary output of the algorithm is allowed.)

\subsection{Quantum MAX-CUT}

As mentioned above, there is a general theory of hardness of approximation for CSPs, including PCP Theorems, the Unique Games Conjecture, and Raghavendra's sharp inapproximability result \cite{ragh08}.  However, these results do not directly apply to their quantum analogues.  The quantum analogue of a constraint satisfaction problem is a local Hamiltonian problem.  In this context, there is a conjectured but still unproven form of the PCP Theorem \cite{aharonov13}, so the hardness of approximation for these problems is less understood.  (Despite recent progress such as \cite{berg25}, the Quantum PCP Conjecture is still unresolved.)  So, it is natural to investigate the inapproximability of a relatively simple version of a local Hamiltonian problem, such as Quantum MAX-CUT.

The Quantum MAX-CUT problem is a special case of the local Hamiltonian problem that in some sense generalizes the usual MAX-CUT problem.  The Quantum MAX-CUT problem is QMA-complete \cite{piddock17}, which is the quantum analogue of NP-complete.  As with MAX-CUT, it is natural to ask for approximation algorithms for Quantum MAX-CUT, and to try to prove sharp computational hardness of those algorithms.  Such a hardness result would then be evidence towards a quantum version of the PCP theorem.  Below, we mainly discuss classical algorithms for Quantum MAX-CUT, i.e. we do not discuss quantum algorithms for Quantum MAX-CUT.

The Quantum MAX-CUT problem was formally described in 2019 \cite{ghar19}, though it is a natural maximization variant of the anti-ferromagnetic Heisenberg XYZ model, a classic family of QMA-complete \cite{cubitt16,piddock17} 2-local Hamiltonians first investigated nearly a century ago by Heisenberg \cite{heis28} which models magnetic systems where interacting particles have opposing spins. 

The Quantum MAX-CUT problem itself has some efficient approximation algorithms \cite{apte25,bakshi26}, though there does not appear to be a conjectured optimal algorithm.

For the product state Quantum MAX-CUT problem, a conjecturally optimal efficient algorithm was found in 2014 by \cite{BrietOliveiraVallentin} with a multiplicative approximation of
\[
\alpha_{\rm BOV}\colonequals 0.9563372685\ldots
\]
Then, \cite{hwang21} showed a matching NP-hardness result conditional on the Unique Games Conjecture, assuming an unproven $S^{k-1}$-valued Borell inequality \cite{borell85} with $k=3$.

In this manuscript, we prove the conjecture from \cite{hwang21} for all correlation parameters in $[-.5843,.5843]$ when $k=3$, and all correlations in $[-3/5,3/5]$ when $k>3$.  The $k=3$ result is then sufficient to deduce sharp hardness of approximation for the algorithm of \cite{BrietOliveiraVallentin}, assuming the Unique Games Conjecture.

\begin{theorem}[$S^{k-1}$-Valued Borell Inequality, Informal]\label{thm4}
Let $n\geq k\geq3$.  If $k=3$, fix a correlation $-.5843\leq\rho<0$.  If $k>3$, fix a correlation $-3/5\leq\rho<0$.  Then the function $f\colon\R^n\to S^{k-1}$ with smallest noise stability with correlation $\rho$ is $f(x)\colonequals (x_1,\ldots,x_k)/\sqrt{x_1^2 + \cdots + x_k^2}$ for all $x=(x_1,\ldots,x_n)\in\R^n$ where $x_1^2+\cdots+x_k^2\neq0$.  (The function $f$ can be defined arbitrarily on the set $x_1^2+\cdots+x_k^2=0$.)

For any correlation $0<\rho\leq.5843$ (or $0<\rho\leq3/5$ when $k>3$), this function $f$ maximizes noise stability among all mean zero functions (with respect to the Gaussian measure).

Moreover, up to rotations and measure zero changes, this is the only optimizing function.
\end{theorem}

The case $\rho=0$ is uninteresting, since the noise stability of a function at correlation zero is just the squared length of its Gaussian expected value.

As shown in \cite{hwang21}, Theorem \ref{thm4} with $k=3$ allows us to conclude sharp hardness for Quantum MAX-CUT.

\begin{theorem}[\embolden{Sharp Hardness for Product State Quantum MAX-CUT}]\label{thm5}
Assume the Unique Games Conjecture.  Then, for any $\epsilon>0$, it is NP-hard to approximate the product state Quantum MAX-CUT problem within a multiplicative factor $\alpha_{\rm BOV}+\epsilon$.
\end{theorem}

Define $\alpha_{\rm BOV,k}$ in \eqref{eq:alpha-bov} below.  As shown in \cite{hwang21}, Theorem \ref{thm4} with $k\geq 3$ allows us to conclude sharp hardness for rank-$k$ MAX-CUT.
\begin{theorem}[Sharp Hardness for rank-$k$ MAX-CUT]\label{thmk}
Let $k\geq3$.  Assume the Unique Games Conjecture.  Then, for any $\epsilon>0$, it is NP-hard to approximate the rank-$k$ MAX-CUT problem within a multiplicative factor $\alpha_{\rm BOV,k}+\epsilon$.
\end{theorem}

The product state of Quantum MAX-CUT is a special case of Quantum MAX-CUT.  More specifically, \cite[Theorem 11.3]{hwang21} shows that Theorem \ref{thm4} implies the same hardness of approximation as in Theorem \ref{thm5} to the Quantum MAX-CUT problem itself.
\begin{corollary}\label{cor1}
Assuming the Unique Games Conjecture, for every $\epsilon>0$, it is NP-hard to approximate Quantum MAX-CUT within a multiplicative factor $\alpha_{\rm BOV}+\epsilon$.
\end{corollary}

Efficient algorithms for Quantum MAX-CUT have approximation guarantees smaller than $\alpha_{\rm BOV}$ at present, so it is unclear if Corollary \ref{cor1} is sharp or not.  For example, there is a $.614$ approximation algorithm in the unweighted case \cite{bakshi26} and a $.611$ approximation algorithm in the weighted case \cite{apte25}, which are both quite far from the $\alpha_{\rm BOV}\approx .956$ hardness result from Corollary \ref{cor1}.  Still, Quantum MAX-CUT has an unconditional NP-hardness of approximation result \cite{piddock25} for some unspecified constant less than one.

\subsection{Discussion of the Proofs}

At a high level, the proofs of the main results Theorems \ref{thm3} and \ref{thm5} proceed along the same lines as \cite{heilman25} and \cite{Heilman2023} by sharpening various estimates, though the details are fairly different.  We first summarize the proof of the sphere-valued Borell inequality, Theorem \ref{thm4}, which implies Theorem \ref{thm5}.

\begin{itemize}
\item As in \cite{heilman25}, split the noise stability of a function $f$ on $\R^n$ into two terms, one of which is the (bilinear) average of the average value of $f$ on spheres centered at the origin, and the remaining part.  (See Lemma \ref{lem:shell}.)
\item Whereas \cite{heilman25} estimates these two terms separately using a sequence of integral estimates, we instead write one single integral estimate as an inequality for an infinite-dimensional matrix in the radial Laguerre basis on $(0,\infty)$.  (See Proposition \ref{prop:radial-criterion}.)
\item The final estimate then only requires proving this infinite-dimensional matrix has operator norm less than one.  (See Section \ref{sec:certificate} culminating in Proposition \ref{prop:certificateB}.)
\item Since the matrix entries decay exponentially, this can be done by numerically estimating the operator norm of a finite $20\times 20$ matrix with positive entries, together with an elementary analytic decay bound on the remaining (small) infinite matrix.
\end{itemize}
In other words, the final estimate amounts to bounding the largest eigenvalue of a specific $20\times 20$ symmetric matrix \eqref{eq:A-def} with nonnegative entries.

We now summarize the proof of the three set Standard Simplex Problem, Theorem \ref{thm1} which implies Theorem \ref{thm3}.

\begin{itemize}
\item As in \cite{Heilman2023}, it suffices to consider three subsets of $\R^2$ by \cite{HeilmanTarter2021}.  Then write the noise stability of these sets as a (bilinear) average of restrictions to spheres centered at the origin, and a remaining part.  (See \eqref{kapdef} or \cite[Section 2]{Heilman2023}.)
\item Whereas \cite{Heilman2023} splits these spherical noise stabilities into a few separate terms and combines these separate estimates later on, we instead create one estimate for any pair of radii of spheres, and then integrate it over the radial coordinates to obtain a single integral operator.  (See \eqref{eq:K-integral} and \eqref{eq:endpoint-joint-arc} then \eqref{eq:integrated-star-form}.)
\item The final estimate then requires that the eigenvalues of this integral operator be bounded by $2$ in absolute value.  (See Lemma \ref{speclem} and \eqref{finaleq}.)
\item It further suffices to bound two suitable Hilbert-Schmidt norms of these integral operators by $2$. (It is possible to perform operator norm estimates for an infinite matrix in the radial Laguerre basis, but bounding the Hilbert-Schmidt norm is a bit simpler, so we take that approach.) 
\item After a change of variables to remove a singularity, these Hilbert-Schmidt bounds can be verified numerically, by explicitly upper bounding some bounded continuous integrals on $[0,5]^2$, combined with an analytic tail bound (since the integrals decay exponentially at infinity) on $[0,\infty)^2\setminus[0,5]^2$.  (See Lemma \ref{lem:certificate}.)  
\end{itemize}

The strategy for Theorem \ref{thm1} hides several subtle issues that were found in \cite{heilman16} and \cite{Heilman2023}.  For example, Theorem \ref{thm1} cannot hold when unequal measure restrictions are imposed on the sets.  So, the proof of Theorem \ref{thm1} must somehow get around this issue by using some property of the three 120 degree sectors that does not hold for partitions into non-congruent sectors.  This is achieved by first performing a rearrangement on circles centered at the origin (as was done already in \cite{Heilman2023}).  The other main technical issue (also circumvented in \cite{Heilman2023}) is that when you decompose the noise stability into different pieces, those different pieces may no longer have the same optimizers (of three 120 degree sectors).  In the present work, this issue is dealt with by integrating in the radial direction and bounding the operator norm of the whole operator, instead of decomposing into pieces and estimating each piece separately.

The need for this new strategy arose since previous approaches to noise stability inequalities (for a single set) did not appear to generalize to the setting of Theorem \ref{thm1}.  Such strategies included: semigroup monotonicity \cite{ledoux94}; induction from discrete settings \cite{bobkov97}; rearrangement \cite{burchard01}; differential inequalities \cite{mossel15}; stochastic calculus \cite{eldan15}.  Our approach instead uses: (i) calculus of variations (to reduce the ambient dimension \cite{HeilmanTarter2021}), (ii) then rearrangement on spheres, (iii) Fourier analytic decompositions, (iv) and finally, reassembled operator norm estimates.

As we can see, the approaches to Theorems \ref{thm1} and \ref{thm4} are related but distinct.  The step for averaging over spheres is common to both approaches, though we choose slightly different strategies for bounding the operator norms.  The details of the arguments are also fairly different.  Theorem \ref{thm1} requires a rearrangement argument on spheres that is not needed in Theorem \ref{thm4}.  Hence, Theorem \ref{thm4} is in some sense easier than Theorem \ref{thm1}.  Another difference between Theorems \ref{thm1} and \ref{thm4} is the Fourier analytic formulas, since the latter uses Fourier analysis on $S^{k-1}$ with $k\geq3$ and with exponential measures on $(0,\infty)$, while the former uses Fourier analysis on $S^1$.  Lastly, it is easier in Theorem \ref{thm4} to obtain a full range of positive and negative correlation values, while this is a more difficult task in Theorem \ref{thm1} due to a lack of any obvious monotonicity with respect to the correlation parameter.

The sphere averaging idea originated in \cite{hwang21}, although their attempted proof of Theorem \ref{thm4} was erroneous.  Weaker versions of Theorems \ref{thm1} and \ref{thm4} appeared in \cite{Heilman2023} and \cite{heilman25}, respectively.  The innovation of the present paper is to obtain suitably sharp estimates in order to obtain the appropriate correlation quantity required for the inapproximability results in Theorems \ref{thm3} and \ref{thm5}.

\subsection{Possible Future Work}

\begin{itemize}
\item As mentioned above, our Standard Simplex and Plurality is Stablest results are proven for correlations in $[-1/2,2/5]$.  Some small improvements beyond these values seem possible, but it seems difficult to make the argument work simultaneously for all correlations in $(-1,1)$.  Perhaps it is an accident that a single argument works this way in the case of a single set in Euclidean space \cite{mossel10}.  It is well understood that correlations close to $0$ have a more Fourier analytic content, while correlations close to one have more isoperimetric content.  Our argument combines Fourier analytic arguments with functional analytic bounds, in contrast to the various other quite different approaches to the single set case.
\item The inapproximability result Theorem \ref{thm:max-2lin-3-hardness} for MAX-2-LIN(3) can be strictly improved by resolving the three candidate Plurality is Stablest Conjecture for all positive correlations.  The optimal choice of correlation is around $.615$, which seems difficult to achieve at present.
\item It should be possible to write stronger stability estimates of Theorems \ref{thm1} and \ref{thm4}, though for the sake of brevity we did not pursue this approach.  A stability estimate says: if you are close to optimizing the functional value, then you are close in the domain to the optimizer.
\item The analogue of Theorem \ref{thm1} for four or more sets, and Theorem \ref{thm2} for four or more candidates both remain largely open.  Some of the steps in the present paper carry over readily to this setting, but some do not.  There are a few significant technical issues that arise in these settings, most notably the lack of rearrangement.  The first step of Theorem \ref{thm1} is to rearrange the sets on spheres centered at the origin, which amounts to a one-dimensional rearrangement on $S^1$ into three circular arcs.  Moreover, it does not matter if these arcs intersect each other or not (and a geodesic ball and a spherical arc are the same on $S^1$, which is good since rearrangement on $S^1$ naturally moves toward geodesic balls).  For four sets in $\R^3$, the analogous step would be a two-dimensional rearrangement of four sets in $S^2$ into four spherical triangles.  On $S^2$, a geodesic ball is not the same as a spherical triangle, so the most straightforward rearrangement inequalities on $S^2$ cannot apply here.  There does not appear to be any result in the literature or any general strategy that allows for rearrangement of four disjoint sets on the sphere into four disjoint spherical triangles.  This is the main obstruction to progress.  \cite{Heilman2022} provides an alternative strategy to circumvent this issue.
\item It should be possible to use our results to construct integrality-gap
instances showing that the Frieze--Jerrum semidefinite program for
MAX-3-CUT attains its worst-case behavior.
\item Theorem \ref{thm5} addresses the product state of Quantum MAX-CUT.  There is also a generalization of that problem known as (the product state of) Quantum MAX-k-CUT \cite{carlson23}.  It could be interesting to pursue a sharp hardness result in that more general setting.  Theorem \ref{thm5} only required $S^2$-valued functions in Theorem \ref{thm4}.  An analogue of Theorem \ref{thm5} for Quantum MAX-k-CUT may require $S^{k-1}$-valued functions, which is one reason we proved Theorem \ref{thm4} for $S^{k-1}$-valued functions for all $k\geq3$ (besides Theorem \ref{thmk}).  However, it is not clear exactly what inequality is needed for Quantum MAX-k-CUT (since extra constraints may need to be placed on the function range), or what the ``target correlation value'' should be (i.e. what is the analogue of $\rho_{\rm BOV,k}$ in \eqref{eq:alpha-bov}).
\end{itemize}

\subsection{GitHub Link}

Associated codes for Theorems \ref{thm1} and \ref{thm4} appear at the following GitHub link: \url{https://github.com/sheilman77/cut_bounds}.

\section{Part I. Quantum MAX-CUT Preliminaries}

Let $n\geq k\geq3$.  Let $\gamma_n(x)\colonequals (2\pi)^{-n/2}e^{-\|x\|^{2}/2}$ for all $x\in\R^{n}$ with $\|x\|\colonequals(\sum_{i=1}^{n}x_i^2)^{1/2}$ for all $x=(x_1,\ldots,x_n)\in\R^{n}$.  For any
$-1<\rho<1$ and for any bounded measurable $f\colon\R^{n}\to\R^{k}$, let $T_\rho$ be the Ornstein--Uhlenbeck operator
\[
  T_\rho f(x)
  \colonequals
  \int_{\R^{n}}f(\rho x + y\sqrt{1-\rho^{2}})\gamma_{n}(y)dy,\qquad\forall\,x\in\R^{n}.
\]

For a bounded measurable $f\colon\R^n\to\B^k$, where
$\B^k=\{x\in\R^k\colon\norm{x}\leq1\}$, define the noise stability with correlation parameter $\rho$ to be
\[
  \Stab_\rho(f)
  \colonequals\int_{\R^n}\ip{f(x)}{T_\rho f(x)}\,\gamma_n(x)\dd x.
\]
Equivalently, $\Stab_\rho(f)=\E\ip{f(X)}{f(Y)}$ when $X,Y$ are
$\rho$-correlated standard Gaussians.

For any $n\geq k\geq 3$, put 
\[
  f_{\mathrm{opt},k}(x)
  \colonequals\frac{(x_1,\ldots,x_k)}{(x_1^2+\cdots+x_k^2)^{1/2}},\qquad\forall\,x=(x_1,\ldots,x_n)\in\R^n\text{ with }x_1^2+\cdots+x_k^2\neq0.
\]
Its noise stability is expressed using hypergeometric functions as \cite[Theorem~7.15]{hwang21}
\begin{equation}\label{eq:Fstar}
 F^*(k,\rho)
 \colonequals\Stab_\rho(f_{\mathrm{opt},k})\\
 =\frac{2}{k}
 \left(\frac{\Gamma((k+1)/2)}{\Gamma(k/2)}\right)^2
 \rho\cdot\,{}_2F_1\!\left(\frac12,\frac12;
 \frac{k}{2}+1;\rho^2\right).
 \end{equation}

Theorem \ref{thm4} follows from the following two theorems.

\begin{theorem}[\embolden{Negative vector-valued Borell inequality}]
\label{thm:negative-borell}
Let $n\geq k\geq3$, let $f\colon\R^n\to\B^k$ be measurable, and let
$
  -0.5843\leq\rho<0
$
(or $-3/5\leq\rho<0$ if $k>3$).
Then
\begin{equation}
\label{eq:negative-borell}
  \Stab_\rho(f)\geq F^*(k,\rho).
\end{equation}
\end{theorem}

\begin{theorem}[\embolden{Positive vector-valued Borell inequality}]
\label{thm:positive-centered}
Let $n\geq k\geq3$, let $f\colon\R^n\to\B^k$ be measurable with
$\int_{\R^{n}} f\dd\gamma_n=0$.  Let $0<\eta\leq0.5843$ (or $0<\eta\leq3/5$ if $k>3$).  Then
\begin{equation}
\label{eq:positive-centered}
  \Stab_\eta(f)\leq F^*(k,\eta).
\end{equation}
\end{theorem}

Let $\rho_{\rm BOV,k}$ be the value of $\rho$ achieving the infimum below.  That is
\begin{equation}
\label{eq:alpha-bov}
  \alpha_{\mathrm{BOV},k}
  \colonequals\inf_{-1<\rho<0}\frac{1-F^*(k,\rho)}{1-\rho}
    =\frac{1-F^*(k,\rho_{\mathrm{BOV},k})}{1-\rho_{\mathrm{BOV},k}}.
\end{equation}
As shown in \cite{BrietOliveiraVallentin,hwang21}, when $k=3$, we have $\rho_{\rm BOV, 3}=-0.5842676623\ldots$ and $\alpha_{\rm BOV, 3}=0.9563372685\ldots$.
Theorem \ref{thm:negative-borell} together with \cite[Theorem 11.3]{hwang21} proves Theorem \ref{thm5} since
\begin{equation}\label{rhobov3}
-.5843<\rho_{\rm BOV,3}<0.
\end{equation}

\subsection{Quantum MAX-CUT}

Below we describe the Quantum MAX-CUT problem by analogy with MAX-CUT.  When $M$ is a $2\times 2$ matrix and $j$ is a positive integer, we denote
$$M^{\otimes j}\colonequals\underbrace{M\otimes \cdots\otimes M}_{j\rm\ times}.$$
If $n$ is a positive integer and $1\leq j\leq n$, denote
$$Z_{j}\colonequals I_{2}^{\otimes (j-1)}\otimes \begin{pmatrix} 1 & 0 \\ 0 & -1\end{pmatrix} \otimes I_{2}^{\otimes (n-j)},\qquad\forall\,1\leq j\leq n,
\qquad I_{2}\colonequals\begin{pmatrix} 1 & 0 \\ 0 & 1\end{pmatrix}.$$
The (weighted) \textbf{MAX-CUT} problem can be stated as: given $w\colon\{1,\ldots,n\}^{2}\to[0,\infty)$ satisfying $w_{ij}=w_{ji}$ and $w_{ii}=0$ for all $1\leq i,j\leq n$, compute the following quantity  \cite{ghar19,hwang21}
$$\max_{u\in(\C^{2})^{\otimes n}\colon\norm{u}\leq1}u^{*}\Big(\sum_{i,j=1}^{n}w_{ij}(I_{2}^{\otimes n} - Z_{i}Z_{j})\Big) u.$$
Define now

$$X_{j}\colonequals I_{2}^{\otimes (j-1)}\otimes \begin{pmatrix} 0 & 1 \\ 1 & 0\end{pmatrix} \otimes I_{2}^{\otimes (n-j)},\qquad\forall\,1\leq j\leq n,$$
$$Y_{j}\colonequals I_{2}^{\otimes (j-1)}\otimes \begin{pmatrix} 0 & -\sqrt{-1} \\ \sqrt{-1} & 0\end{pmatrix} \otimes I_{2}^{\otimes (n-j)},\qquad\forall\,1\leq j\leq n.$$

The \textbf{Quantum MAX-CUT} problem is \cite{ghar19,hwang21}: given $w\colon\{1,\ldots,n\}^{2}\to[0,\infty)$ satisfying $w_{ij}=w_{ji}$ and $w_{ii}=0$ for all $1\leq i,j\leq n$, compute the following quantity
$$\max_{u\in(\C^{2})^{\otimes n}\colon \norm{u}\leq 1}u^{*}\Big(\sum_{i,j=1}^{n}w_{ij}(I_{2}^{\otimes n}- X_{i} X_{j} - Y_{i} Y_{j} - Z_{i} Z_{j})\Big) u.$$
The \textbf{product state of Quantum MAX-CUT} is the more restricted optimization problem of computing
$$\max_{\substack{u=u_{1}\otimes\cdots\otimes u_{n}\colon\\ u_{i}\in \C^{2},\,\norm{u_{i}}\leq1,\,\forall\,1\leq i\leq n}}u^{*}\Big(\sum_{i,j=1}^{n}w_{ij}(I_{2}^{\otimes n} - X_{i} X_{j} - Y_{i} Y_{j} - Z_{i} Z_{j})\Big) u.$$

\section{The weighted radial reduction}
\label{sec:reduction}

Fix $0<\eta<1$.  Following the shell decomposition in
\cite{heilman25}, we first prove the centered positive inequality in
dimension $k$.  Let $X,Y\in\R^k$ be $\eta$-correlated standard Gaussians
and write $R=\norm{X}$, $S=\norm{Y}$, $U=X/R$, and $V=Y/S$, away from null
sets.  Let $I_m$ denote the modified Bessel function \eqref{imdef}.  Conditional on $R=r$ and $S=s$, the angular noise operator on $S^{k-1}$
has first nonconstant eigenvalue \cite[Eq. 17]{heilman25} \cite[Corollary 4.9]{hwang21}
\begin{equation}
\label{eq:lambda}
 \lambda_{\eta,k}^{r,s}
  =\frac{I_{k/2}(z)}{I_{k/2 - 1}(z)},
  \qquad z=\frac{\eta rs}{1-\eta^2}.
\end{equation}
For any bounded measurable $f\colon\R^n\to\B^k$, let $m\colon(0,\infty)\to\B^k$ be defined by
\begin{equation}\label{mdef}
m(r)=m_f(r)\colonequals\int_{S^{n-1}}f(r\theta)\,\dd\sigma(\theta),\qquad\forall\,0<r<\infty,
\end{equation}
where $\sigma$ is normalized surface measure on $S^{n-1}$ (so that $\sigma(S^{n-1})=1$).  (Note that $m(r)$ is well-defined for a.e. $r\in(0,\infty)$ by Fubini's Theorem.)

\begin{lemma}[Ball-valued shell estimate]
\label{lem:shell}
For every measurable $f\colon\R^k\to\B^k$,
\begin{equation}
\label{eq:shell}
  \Stab_\eta(f)-F^*(k,\eta)
  \leq
  \E\ip{m(R)}{m(S)}
  -\E\!\left[\lambda_{\eta,k}^{R,S}\norm{m(R)}^2\right].
\end{equation}
\end{lemma}
\begin{proof}
This follows from \cite[Eq. 5]{heilman25} and (the proof of) \cite[Lemma 2.2]{heilman25}, the only change being that $q(r)\colonequals\int_{S^{k-1}}\|f(r\theta)\|^2\dd\sigma(\theta)-\|m(r)\|^2\leq 1-\|m(r)\|^2$ in \cite[Lemma 2.1]{heilman25}.
\end{proof}

Let $G\in\R^{k}$ be a standard Gaussian random vector.  Let $s>0$.  Define
\begin{equation}
\label{eq:mk-def}
 m_k
 \colonequals\E\norm{G}
 =\sqrt2\,\frac{\Gamma((k+1)/2)}{\Gamma(k/2)},
\end{equation}
\begin{equation}
\label{eq:Y-def}
 \delta=\delta_{\eta}(s)
 \colonequals\frac{\eta s}{\sqrt{1-\eta^2}},
 \qquad
 \mathcal Y_{\eta,k}(s)
 \colonequals
 \frac{m_k\delta}{k}\cdot\,
 {}_1F_1\!\left(\frac12;\frac{k}{2}+1;-\frac{\delta^2}{2}\right).
\end{equation}

\begin{lemma}[Exact conditional angular weight]
\label{lem:Y-exact}
For every $s>0$,
\begin{equation}
\label{eq:Y-exact}
 \E\!\left[\lambda_{\eta,k}^{R,S}\mid S=s\right]
 =\mathcal Y_{\eta,k}(s).
\end{equation}
Consequently,
\begin{equation}
\label{eq:Y-integrated}
 \E\!\left[\lambda_{\eta,k}^{R,S}\norm{m(R)}^2\right]
 =\E\!\left[\mathcal Y_{\eta,k}(R)\norm{m(R)}^2\right].
\end{equation}
\end{lemma}

\begin{proof}
Denote $e_1=(1,0,\ldots,0)\in\R^k$.  Since $(X,Y)$ is rotation invariant, $\E[\lambda_{\eta,k}^{R,S}\mid S=s]=\E[\lambda_{\eta,k}^{R,S}\mid Y=se_1]$.  Conditional on $Y=se_1$, $X$ is equal in distribution to $\eta se_1+\sqrt{1-\eta^2}G=\sqrt{1-\eta^2}(\delta e_1 + G)$ by \eqref{eq:Y-def}.  Conditional on $R=r$, $U=X/\|X\|$ then has the von Mises-Fisher distribution with known expected value $\E[U\mid R=r,Y=se_1]=\lambda_{\eta,k}^{r,s}e_1$.  Consequently, $\E[U_1\mid R=r,Y=se_1]=\lambda_{\eta,k}^{r,s}$.  Taking the conditional expectation over $R$,
\[
\E[\lambda_{\eta,k}^{R,S}\mid S=s]
=\E[\lambda_{\eta,k}^{R,s}\mid Y=se_1]
=\E[U_1\mid Y=se_1]
=\E\frac{\delta+G_1}{\|\delta e_1 + G\|}.
\]
We compute this quantity as $M_k'(\delta)$ where $M_k(\delta)\colonequals \E\|\delta e_1 +G\|$.  The expected value of the noncentral chi distribution has a known formula $M_k(\delta)=m_k \!\cdot\,{}_1F_1\!\left(
 -\frac12;\frac{k}{2};-\frac{\delta^2}{2}\right)$ in terms of the confluent hypergeometric function.  Then using a known differentiation formula, $\frac{\dd}{\dd z}{}_1F_1(a;b;z)
=(a/b){}_1F_1(a+1;b+1;z)$, we get
\[
\E[\lambda_{\eta,k}^{R,S}\mid S=s]
=M_k'(\delta)
=m_k(-\delta)(-1/k) \!\cdot\,{}_1F_1\!\left(
 \frac12;\frac{k}{2}+1;-\frac{\delta^2}{2}\right)
 \stackrel{\eqref{eq:Y-def}}{=}Y_{\eta,k}(s).
\]
That is, \eqref{eq:Y-exact} holds.  Equation \eqref{eq:Y-integrated} follows from exchangeability.
\end{proof}

The next bound is the only estimate applied to the exact weight.

\begin{lemma}[Reciprocal bound]
\label{lem:kummer-reciprocal}
Let $k\geq3$, $s>0$. Let $\delta=\delta_\eta(s)$ from \eqref{eq:Y-def}.  Then
\begin{equation}
\label{eq:kummer-reciprocal}
 \frac{1}{\mathcal Y_{\eta,k}(s)}
 \leq
 \sqrt{1+\left(\frac{k}{m_k\delta}\right)^2}
 \leq1+\frac{k}{m_k\delta}.
\end{equation}
\end{lemma}

\begin{proof}
Let
$B$ have distribution $\operatorname{Beta}(1/2,(k+1)/2)$.  Kummer's integral formula gives
\begin{equation}\label{f1id}
 {}_1F_1\!\left(\frac12;\frac{k}{2}+1;-\frac{\delta^2}{2}\right)
 =\E e^{-\delta^2B/2}.
\end{equation}
Let $H$ have the gamma distribution with shape $\alpha=1/2$ and scale
$
 \theta_k=2\left(\frac{m_k}{k}\right)^2.
$
The leading constants at zero in the densities of $B$ and $H$ agree.
Indeed, by \eqref{eq:mk-def},
$
 \theta_k^{-1/2}
 =\frac{\Gamma(k/2+1)}{\Gamma((k+1)/2)}
$.
On $0<t<1$, the density ratio is
\[
 \frac{f_B(t)}{f_H(t)}
 =(1-t)^{(k-1)/2}e^{t/\vartheta_k}.
\]

Its logarithmic derivative,
\[
 \frac1{\theta_k}-\frac{k-1}{2(1-t)},
\]
is strictly decreasing, is positive at zero, and tends to $-\infty$ as
$t\uparrow1$.  Positivity at zero follows from
$m_k^2<k$, by \eqref{eq:mk-def} and Jensen's inequality, so $\theta_k<2/k$, i.e. $1/\theta_k > k/2$.  Thus the
two densities have
one crossing at some $0<t_0<1$: $f_B(t)\geq f_H(t)$ for $0<t<t_0$ and $f_B(t)\leq f_H(t)$ for $t_0<t<1$, and hence
$B$ is stochastically dominated by $H$.  So, using this property and the known formula for the moment generating function of $H$,
\[
{}_1F_1\!\left(\frac12;\frac{k}{2}+1;-\frac{\delta^2}{2}\right)
\stackrel{\eqref{f1id}}{=}
 \E e^{-\delta^2B/2}
 \geq\E e^{-\delta^2H/2}
 =\left(1+\frac{m_k^2\delta^2}{k^2}\right)^{-1/2},\qquad\forall\,\delta>0.
\]
Multiplying by $m_k\delta/k$ and using \eqref{eq:Y-def}
proves the first inequality.  The second is
$\sqrt{1+x^2}\leq1+x$ valid for any $x\geq0$.
\end{proof}

We now isolate the radial quadratic form.  The random variable $R^2/2$ has law \begin{equation}\label{eq:mu3}
\dd\nu_k(t)\colonequals\frac{t^{(k/2)-1}e^{-t}}{\Gamma(k/2)}\,\dd t.
\end{equation}

Let
\begin{equation}
\label{ljdef}
 \ell_j^{(k)}(t)
 \colonequals
 \left(\frac{j!\Gamma(k/2)}{\Gamma(k/2 + j)}\right)^{1/2}
 L_j^{(k/2)-1}(t),
 \qquad j\geq0,
\end{equation}
where $L_j^{(\alpha)}(t)\colonequals\frac{t^{-\alpha}e^t}{j!}\frac{d^j}{dt^j}(e^{-t}t^{j+\alpha})$, $\forall$ $t\in\R$.  Then $(\ell_j^{(k)})_{j\geq0}$ is an orthonormal basis of
$L^2(\nu_k)$.  The radial Ornstein--Uhlenbeck operator satisfies
\begin{equation}
\label{tel}
 T_\eta\ell_j^{(k)}(\norm{x}^2/2)
 =\eta^{2j}\ell_j^{(k)}(\norm{x}^2/2),
\end{equation}
since for a smooth function $h\colon\R\to\R$ with $t\colonequals\|x\|^2/2$, we have $\Delta h(t)-\langle x,\nabla h(t)\rangle = 2[h''(t)t + (k/2-t)h'(t)]$, and the generalized Laguerre polynomials with parameter $\alpha=k/2-1$ satisfy $t\ell_j''+(\alpha+1-t)\ell_j'+j\ell_j=0$, so $(\Delta-\langle x,\nabla\rangle)\ell_j(t)=-2j\ell_j(t)$, hence \eqref{tel} follows.  Finally, the orthonormality with respect to $\dd\nu_k(t)$ follows from known orthogonality and the identity $\int_{0}^{\infty}t^{\alpha}e^{-t}(L_j^{(\alpha)}(t))^2 \dd t=\Gamma(j+\alpha+1)/j!$.

Define the positive semidefinite matrix
$A_{\eta}^{(k)}=((A_{\eta}^{(k)})_{ij})_{i,j\geq1}$ by
\begin{equation}
\label{eq:A-def}
 (A_{\eta}^{(k)})_{ij}
 =\eta^{i+j}\int_0^\infty
 \frac{\ell_i^{(k)}(t)\ell_j^{(k)}(t)}
 {\mathcal Y_{\eta,k}(\sqrt{2t})}\,\dd\nu_k(t).
\end{equation}

Here $\norm{A_\eta^{(k)}}$ denotes the supremum of the operator norms of its
finite principal sections; Proposition \ref{prop:certificate} below shows that this
quantity is finite.

\begin{lemma}[Weighted radial analysis]
\label{lem:weighted-analysis}
For every measurable $m\colon(0,\infty)\to\R^k$ for which the right-hand side is
finite,
\begin{equation}
\label{eq:weighted-analysis}
  \sum_{j\geq1}\eta^{2j}
    \left\|\int_0^\infty m(\sqrt{2t})\ell_j^{(k)}(t)\,\dd\nu_k(t)\right\|^2
  \leq \norm{A_\eta^{(k)}}\int_0^\infty
     \mathcal Y_{\eta,k}(\sqrt{2t})\|m(\sqrt{2t})\|^2\,\dd\nu_k(t).
\end{equation}
\end{lemma}

\begin{proof}
It suffices to consider scalar $m$.  For every finitely supported sequence
$c=(c_j)_{j\geq1}$, Cauchy--Schwarz together with dividing and multiplying by $\sqrt{\mathcal Y_{\eta,k}(\sqrt{2t})}$ gives
\begin{align*}
 &\left|\sum_{j\geq1}c_j\eta^j
   \int_{0}^{\infty} m(\sqrt{2t})\ell_j^{(k)}(t)\,\dd\nu_k(t)\right|^2
   =\left|
   \int_{0}^{\infty} m(\sqrt{2t})\sum_{j\geq1}c_j\eta^j\ell_j^{(k)}(t)\,\dd\nu_k(t)\right|^2\\
 &\qquad
 \leq
   \int_{0}^{\infty} |m(\sqrt{2t})|^2\mathcal Y_{\eta,k}(\sqrt{2t})\dd\nu_k(t)
   \cdot\int_{0}^{\infty} \frac{\Big(\sum_{j\geq1}c_j\eta^j\ell_j(t)\Big)^2}{\mathcal Y_{\eta,k}(\sqrt{2t})}\,\dd\nu_k(t)\\
  &\qquad\stackrel{\eqref{eq:A-def}}{=}
 \left(\int_{0}^{\infty} \mathcal Y_{\eta,k}(\sqrt{2t})|m(\sqrt{2t})|^2\,\dd\nu_k(t)\right)
 \sum_{i,j\geq1}c_ic_j(A_\eta^{(k)})_{ij}\\
 &\qquad\leq \norm{A_\eta^{(k)}}\norm{c}_{\ell_2}^2
 \int_{0}^{\infty} Y_{\eta,k}(\sqrt{2t})|m(\sqrt{2t})|^2\,\dd\nu_k(t).
\end{align*}
Taking the supremum over $\norm{c}_{\ell^2}=1$ proves the scalar case of
\eqref{eq:weighted-analysis}.  Applying it componentwise proves the
vector-valued case.
\end{proof}

\begin{proposition}[Radial criterion]
\label{prop:radial-criterion}
If $\norm{A_\eta^{(k)}}<1$, then \eqref{eq:positive-centered} holds for every
mean-zero measurable map $f\colon\R^n\to\B^k$ and every $n\geq k\geq 3$.
\end{proposition}

\begin{proof}
First let $n=k$ and let $m$ be the spherical mean of $f$.  Since
$\int_{\R^k} f\,\dd\gamma_k=0$, the constant Laguerre coefficient of $m$ vanishes.
Define $m=m_f$ from \eqref{mdef}.  For any $j\geq1$, write
\[
  \widehat m_j
  \colonequals\int_0^\infty m(\sqrt{2t})\ell_j^{(k)}(t)\,\dd\nu_k(t).
\]
Recall that the modified Laguerre polynomials $(\ell_j^{(k)}(t))_{j\geq0}$ with $t=\|x\|^2/2$ form a complete orthonormal basis of radial functions on $\R^k$, since e.g. any rotation invariant polynomial of $x$ is a polynomial in $t$.  So, \eqref{tel} together with $\widehat m_0=0$ give
\[
  \E\ip{m(R)}{m(S)}
  =\sum_{j\geq1}\eta^{2j}\norm{\widehat m_j}^2.
\]
Combine Lemmas \ref{lem:shell}, \ref{lem:Y-exact} and \ref{lem:weighted-analysis} with the assumption $\|A_\eta\|<1$ to obtain
\begin{equation}\label{stfsineq}
  \Stab_\eta(f)-F^*(k,\eta)
  \leq\bigl(\norm{A_\eta}-1\bigr)
  \E\!\left[\mathcal Y_{\eta,k}(R)\norm{m(R)}^2\right]\leq0.
\end{equation}
For $n>k$, \cite[Theorem~6.1]{hwang21} says that the
maximum among mean-zero maps $\R^n\to\B^k$ is attained and that every
maximizer depends, after an orthogonal change of coordinates, on at most
$k$ coordinates.  The $k$-dimensional result therefore proves the claim
in every dimension $n>k$.
\end{proof}

\section{A uniform spectral certificate}
\label{sec:certificate}

We now prove $\norm{A_\eta^{(k)}}<1$ uniformly through $\eta=0.5843$ when $k=3$ (and through $\eta=3/5$ with $k>3$).  Instead of bounding $\norm{A_\eta^{(k)}}$ directly, we bound the operator norm of a slightly larger operator $V_\eta^{(k)}$ whose entries have more explicit (hence more tractable) formulas.

Lemma \ref{lem:kummer-reciprocal} with \eqref{eq:mk-def} and \eqref{eq:Y-def} give
\begin{equation}\label{yup}
 \frac1{\mathcal Y_{\eta,k}(\sqrt{2t})}
 \leq1+\beta_{\eta,k}t^{-1/2},\quad\forall\,t>0,
 \qquad
 \beta_{\eta,k}
 \colonequals\frac{k\sqrt{1-\eta^2}}{2\eta}
 \frac{\Gamma(k/2)}{\Gamma((k+1)/2)}.
\end{equation}
Define
\begin{equation}
\label{eq:V-operator}
 V_\eta^{(k)}
 \colonequals
T_\eta^{1/2}M_{1+\beta_{\eta,k}t^{-1/2}}T_\eta^{1/2}.
\end{equation}

By \eqref{yup} and \eqref{eq:Y-def}, we have the operator inequality $M_{1/\mathcal Y_{\eta,k}(\sqrt{2t})}
\leq M_{1+\beta_{\eta,k}t^{-1/2}}$.  From this we will conclude the operator inequality $A_\eta^{(k)}\preceq V_\eta^{(k)}$ on the nonconstant radial subspace.  Here $M_w$ denotes multiplication by $w$.  Indeed, for any $i,j\geq1$,
$$\langle \ell_i, T_\eta^{1/2}M_{\frac{1}{\mathcal Y_{\eta,k}(\sqrt{2t})}}T_{\eta}^{1/2}\ell_j\rangle
=\langle T_\eta^{1/2}\ell_i, M_{\frac{1}{\mathcal Y_{\eta,k}(\sqrt{2t})}}T_{\eta}^{1/2}\ell_j\rangle
=\int_{0}^{\infty}\eta^{i+j}\frac{\ell_i(t)\ell_j(t)}{\mathcal Y_{\eta,k}(\sqrt{2t})}\dd\nu_k(t)
\stackrel{\eqref{eq:A-def}}{=}(A_\eta)_{ij}.$$
So, if $f=\sum_{j\geq1}c_j\ell_j(t)$, then
\begin{equation}\label{psdineq}
\langle f, (V_\eta^{(k)} - A_\eta^{(k)})f\rangle
=\int_{0}^{\infty}\Big(1+\beta_{\eta,k}t^{-1/2} - \frac{1}{\mathcal Y_{\eta,k}(\sqrt{2t})}\Big)\Big|\sum_{j\geq1}\eta^{j}c_j\ell_j(t)\Big|^2\dd\nu_k(t)
\stackrel{\eqref{yup}\wedge\eqref{eq:Y-def}}{\geq}0.
\end{equation}

We next record the matrix of $V_\eta^{(k)}$.  Let $(z)_0\colonequals1$, $(z)_j\colonequals z(z+1)\cdots(z+j-1)$ for any integer $j\geq1$.  Let
\begin{equation}
\label{eq:sc-def}
 s_j^{(k)}
 \colonequals\left(\frac{j!\Gamma(k/2)}{\Gamma(k/2+j)}\right)^{1/2},
 \qquad
 c_j\colonequals\frac{(1/2)_j}{j!},
 \qquad
 d_j^{(k)}\colonequals\frac{((k-1)/2)_j}{j!}.
\end{equation}
The Laguerre connection formula \cite[Eq.~18.18.18]{Olver},
\[
 L_j^{(k/2-1)}(t)
 =\sum_{m=0}^jc_{j-m}L_m^{((k-3)/2)}(t),
\]
and orthogonality of $L_j^{((k-3)/2)}$ with respect to
$t^{(k-3)/2}e^{-t}\dd t$ on $(0,\infty)$ give
\begin{flalign*}
  \int_0^\infty
  \frac{\ell_i(t)\ell_j(t)}{\sqrt{t}}\,\dd\nu_k(t)
  &\stackrel{\eqref{eq:mu3}}{=}\frac{1}{\Gamma(k/2)}\int_{0}^{\infty}\!\!s_i^{(k)}s_j^{(k)}\sum_{m=0}^{i}\sum_{m'=0}^{j}c_{i-m}L_{m}^{(\frac{k-3}{2})}(t)c_{j-m'}L_{m'}^{(\frac{k-3}{2})}(t)t^{(k-3)/2}e^{-t}\dd t\\
  &=\frac{1}{\Gamma(k/2)}s_i^{(k)} s_j^{(k)}
    \sum_{m=0}^{\min(i,j)}c_{i-m}c_{j-m}\frac{\Gamma(m+(k-1)/2)}{m!}.
\end{flalign*}
\begin{equation}
\label{eq:V-matrix}
 \begin{aligned}
 (V_\eta^{(k)})_{ij}
 &=\langle \ell_i, T_\eta^{1/2}M_{1+\beta_{\eta,k}t^{-1/2}}T_{\eta}^{1/2}\ell_j\rangle
=\langle T_\eta^{1/2}\ell_i, M_{1+\beta_{\eta,k}t^{-1/2}}T_{\eta}^{1/2}\ell_j\rangle\\
 &\stackrel{\eqref{yup}}{=}\eta^{2i}\delta_{ij}
 +\frac{k}{k-1}\sqrt{1-\eta^2}\,
 \eta^{i+j-1}s_i^{(k)}s_j^{(k)}\cdot
 \sum_{m=0}^{\min(i,j)}
 c_{i-m}c_{j-m}d_m^{(k)},
 \qquad\forall\, i,j\geq0.
 \end{aligned}
\end{equation}

\begin{proposition}[Certified radial bound, $k=3$]
\label{prop:certificate}
For every $0<\eta\leq0.5843$, on the subspace $\mathrm{span}(1)^{\perp}$ perpendicular to constant functions, we have the operator norm bounds
\begin{equation}
\label{eq:certificate}
  \norm{A_\eta^{(3)}\mid_{\mathrm{span}(1)^{\perp}}}\leq\norm{V_\eta^{(3)}\mid_{\mathrm{span}(1)^{\perp}}}<0.987<1.
\end{equation}
\end{proposition}

\begin{proof}
Set $\eta_0\colonequals0.5843$.  Every entry of \eqref{eq:V-matrix} is nonnegative.
Moreover, $\eta^{2i}$ and
$\eta^{i+j-1}\sqrt{1-\eta^2}$ are increasing for
$0<\eta\leq\eta_0<1/\sqrt2$: indeed, if $p=i+j-1\geq1$, then
$\eta^p\sqrt{1-\eta^2}$ is increasing for
$\eta<\sqrt{p/(p+1)}$.  The norm of a nonnegative self-adjoint
matrix can be tested on nonnegative vectors.  It is therefore enough to upper
bound $\|V_{\eta_0}^{(3)}\|$.

Let $V_{1:20}$ be the leading $20\times20$ block of $V_{\eta_0}^{(3)}$.  We will explicitly describe an approximation $v\in\R^{20}$ to the largest eigenvector of $V_{1:20}$.  For the positive vector
$v$, normalized by $v_1=1$, the following terminating decimals describing $v$ are
interpreted as exact rational numbers:
\[
\begin{array}{rrrrr}
1.000000000000,&0.256858556823,&0.097396033098,&0.043379992735,&0.020897660126,\\
0.010513466380,&0.005433117499,&0.002859305346,
&0.001524852414,&0.000821489732,\\
0.000446145076,&0.000243896114,
&0.000134063610,&0.000074033654,&0.000041046277,\\
0.000022835796,
&0.000012742970,&0.000007129915,&0.000003998799,&0.000002247496,
\end{array}
\]
Directed rational interval arithmetic then gives the following upper bounds on the vector $(V_{1:20}v)_i/v_i$ where $i=1,\ldots,20$:
\[
\begin{array}{ccccc}
0.986730842649 & 0.986730842648 & 0.986730842645 & 0.986730842643 & 0.986730842659 \\
0.986730842616 & 0.986730842679 & 0.986730842649 & 0.986730842823 & 0.986730842081 \\
0.986730842495 & 0.986730841950 & 0.986730841227 & 0.986730846599 & 0.986730840460 \\
0.986730830803 & 0.986730849170 & 0.986730809040 & 0.986730872693 & 0.986730716257
\end{array}
\]
Each displayed entry is rounded upward.  To make the check exact, note that $s_j^2$ and $c_j$ are rational numbers by \eqref{eq:sc-def}.  Thus, at $\eta_0=5843/10000$, every entry of $V_{1:20}$ is algebraic with
only square roots of rational numbers.  Enclosing every square root between
consecutive rationals with denominator $10^{60}$ and carrying out the
remaining operations over $\mathbb Q$ gives the table.  In particular,
\begin{equation}\label{v2bd}
  \max_{1\leq i\leq20}\frac{(V_{1:20}v)_i}{v_i}<0.986731.
\end{equation}
The accompanying file \path{certify_quantum_cut.py} reproduces
these bounds by exact rational arithmetic using only the Python standard
library.  The Collatz--Wielandt upper bound applied to \eqref{v2bd} then gives
\begin{equation}
\label{eq:block-bound}
  \norm{V_{1:20}}<0.986731.
\end{equation}

It remains to control the infinite tail of $V_{\eta_0}^{(3)}$.  Since $s_j\leq1$ and
$0<c_j\leq1$, the diagonal terms satisfy, with $q=\eta_0^2$ and
$\kappa=3\sqrt{1-\eta_0^2}/(2\eta_0)$,
\[
  (V_{\eta_0}^{(3)})_{jj}
  \leq q^j\bigl(1+\kappa(j+1)\bigr),\qquad\forall\,j\geq1.
\]
Therefore the trace of the principal block indexed by $j\geq21$ is at most
\begin{align}
\label{eq:tail}
\frac{q^{21}}{1-q}
    +\kappa\frac{q^{21}(22-21q)}{(1-q)^2}
   <1.150\times10^{-8}\equalscolon\tau.
\end{align}

Recall that the proof of \eqref{psdineq} shows that $V_{\eta_0}^{(3)}$ is symmetric positive semidefinite.  So, we finally combine \eqref{eq:block-bound} and \eqref{eq:tail} with Lemma \ref{lem:block-tail} to get
\[
  \norm{V_{\eta_0}^{(3)}}<.986731+\tau<0.986732<0.987.
\]
This proves \eqref{eq:certificate} for all $0<\eta\leq\eta_0$.
\end{proof}

\begin{lemma}[Positive block tail bound]
\label{lem:block-tail}
Let
$M=\left(\begin{smallmatrix}B&E\\E^*&C\end{smallmatrix}\right)$
be a real symmetric positive semidefinite (finite or infinite) matrix with
$\norm{B}\leq\lambda$ and $\operatorname{tr}C\leq\tau$.  Then
\begin{equation}
\label{eq:block-tail}
 \norm{M}\leq
 \lambda+\tau.
\end{equation}
\end{lemma}
\begin{proof}
Positive semidefiniteness with $r\colonequals-\langle x, Ey\rangle/\langle y, Cy\rangle$ gives
\[
0\leq\binom{x}{ry}
\begin{pmatrix}B&E\\E^*&C\end{pmatrix}
\binom{x}{ry}
=\langle x,Bx\rangle + 2r\langle x, Ey\rangle+r^2\langle y,Cy\rangle
=\langle x, Bx\rangle-\langle x, Ey\rangle^2/\langle y,Cy\rangle.
\]
That is, 
$|\langle x, Ey\rangle|^2\leq\ip{x}{Bx}\ip{y}{Cy}$.  (If it occurs that $\langle y,Cy\rangle=0$, then choosing $r\to\pm\infty$ shows that $\langle x,Ey\rangle=0$ as well.)  Then if $\|\binom{x}{y}\|=1$, using $\|C\|\leq\mathrm{tr}C\leq\tau$, we get
$$\binom{x}{y}\begin{pmatrix} B & E \\ E^*& C\end{pmatrix}\binom{x}{y}\leq \lambda\|x\|^2 + 2\sqrt{\lambda\tau}\|x\|\|y\|
+\tau\|y\|^2
=(\sqrt{\lambda}\|x\|+\sqrt{\tau}\|y\|)^2
\leq(\lambda+\tau)\Big\|\binom{x}{y}\Big\|^2.$$
Since $\|\binom{x}{y}\|=1$, this concludes the bound.
\end{proof}

We now treat the case $k>3$ by certifying two dimension-four operators.  We define the following slight modification of \eqref{eq:V-operator}. 
\begin{equation}
\label{eq:Z-def}
 Z_\eta
 \colonequals
 \eta^2 \cdot T_{\eta}^{1/2}M_{1+\frac{15}{16}\beta_{\eta,4}t^{-1/2}}T_{\eta}^{1/2}.
\end{equation}

\begin{proposition}[$k=4$ Bounds]
\label{prop:k4-certificates}
For $0<\eta\leq3/5$, on the subspace perpendicular to constant functions, we have the operator norm bounds
\begin{equation}
\label{eq:k4-certificates}
 \norm{V_\eta^{(4)}\mid_{\mathrm{span}(1)^\perp}}<0.962,
 \qquad
 \norm{Z_\eta}<0.989.
\end{equation}
\end{proposition}

\begin{proof}
All entries of both matrices are nonnegative and increasing for
$0<\eta\leq3/5<1/\sqrt2$.  It suffices to take $\eta=\eta_0=3/5$.  At $k=4$, by \eqref{eq:sc-def} we have
\begin{equation}\label{sjdj4}
 s_j^{(4)}=\frac1{\sqrt{j+1}},
 \qquad
 d_j^{(4)}=\frac{(3/2)_j}{j!}.
\end{equation}
Consequently, for any $i,j\geq1$, using $k\sqrt{1-\eta^2}/(k-1)=16/15$,
\begin{equation}
\label{eq:V4-endpoint}
 (V_{3/5}^{(4)})_{ij}
 \stackrel{\eqref{eq:V-matrix}}{=}(3/5)^{2i}\delta_{ij}
 +\frac{16}{15}
 \frac{(3/5)^{i+j-1}}{\sqrt{(i+1)(j+1)}}\\
 \cdot
 \sum_{m=0}^{\min(i,j)}
 c_{i-m}c_{j-m}d_m^{(4)},
\end{equation}
whereas, for $i,j\geq0$,
\begin{equation}
\label{eq:Z-endpoint}
 (Z_{3/5})_{ij}
 \stackrel{\eqref{eq:Z-def}}{=}(3/5)^{2i+2}\delta_{ij}
 +\frac{(3/5)^{i+j+1}}{\sqrt{(i+1)(j+1)}}\\
\cdot
 \sum_{m=0}^{\min(i,j)}
 c_{i-m}c_{j-m}d_m^{(4)}.
\end{equation}

Let $V_{1:20}$ be the leading $20\times20$ block of $V_{3/5}^{(4)}$ and let $Z_{0:20}$ be the leading $21\times 21$ block of $Z_{3/5}$.
As in Proposition \ref{prop:certificate}, we will explicitly describe approximations to the largest eigenvectors of $V_{1:20}$ and $Z_{0:20}$, respectively.  For $V_{1:20}$, we use the vector
\[
\begin{array}{rrrrr}
1.000000000000,&0.235502283787,&0.084412886242,&0.036251000101,
&0.017062434237,\\
0.008462067457,&0.004337666976,&0.002274532283,
&0.001212712406,&0.000654921483,\\
0.000357321574,&0.000196590782,
&0.000108918778,&0.000060704275,&0.000034005732,\\
0.000019134287,&
0.000010808471,&0.000006126515,&0.000003483353,&0.000001985998.
\end{array}
\]
As in Proposition \ref{prop:certificate}, for the above vector $v\in\R^{20}$, we computationally upper bound $\max_{i=1,\ldots,20}(V_{1:20}v)_i / v_i$ and obtain an upper bound of $.961$ in this case.  Similarly, for $Z_{0:20}$, we use the vector
\[
\begin{array}{rrrrr}
1.000000000000,&0.195759686570,&0.064856944384,&0.026690658790,
&0.012232805093,\\
0.005957417255,&0.003013602699,&0.001564396228,
&0.000827507987,&0.000444045497,\\
0.000240996163,&0.000132006969,
&0.000072862879,&0.000040478067,&0.000022611648,\\
0.000012691740,
&0.000007153611,&0.000004046971,&0.000002296977,&0.000001307540,\\
0.000000746272.
\end{array}
\]
The entries of \eqref{eq:V4-endpoint} and \eqref{eq:Z-endpoint} involve
only rational numbers and square roots of integers.  Enclosing each square
root by rationals and using the Collatz-Wielandt upper bound,
\begin{equation}
\label{eq:k4-blocks}
 \norm{V_{1:20}}<0.961,
 \qquad
 \norm{Z_{0:20}}<0.98865.
\end{equation}
The largest directed upper bounds for the row ratios are
$0.960805417209$ and $0.988644133222$, respectively.  The bound \eqref{eq:k4-blocks} is verified by \path{certify_rank_k_cut.py}.

We now control the tails of the matrices. Since $c_m\leq1$ by \eqref{eq:sc-def} and
$d_m^{(4)}\leq2m+1$ by \eqref{sjdj4}
\[
 \sum_{m=0}^j c_{j-m}^2d_m^{(4)}
 \leq(j+1)^2.
\]
With $q\colonequals \eta_0^2=(3/5)^2$, the respective tail traces, beginning at $j=21$, are at
most
\begin{align}
\label{eq:k4-tails}
 \tau_V
 &\colonequals
 \frac{q^{21}}{1-q}
 +\frac{16}{9}
 \frac{q^{21}(22-21q)}{(1-q)^2}
 <3.0913\times10^{-8},\\
 \tau_Z
 &\colonequals
 q\frac{q^{21}}{1-q}
 +(3/5)\frac{q^{21}(22-21q)}{(1-q)^2}
 <1.0450\times10^{-8}.
\notag
\end{align}
Combining \eqref{eq:k4-blocks}, \eqref{eq:k4-tails}, and Lemma
\ref{lem:block-tail} gives the two inequalities in
\eqref{eq:k4-certificates}.  
\end{proof}

It remains to show why dimension $k=4$ controls every larger target
dimension $k>4$.

\begin{proposition}[Dimension-four reduction]
\label{prop:k4-reduction}
For any $k\geq5$ and $0<\eta\leq3/5$,
\begin{equation}
\label{eq:k4-reduction}
 \norm{V_\eta^{(k)}\mid_{\mathrm{span}(1)^\perp}}
 \leq
 \max\Big(
 \norm{V_\eta^{(4)}\mid_{\mathrm{span}(1)^\perp}},
 \norm{Z_\eta}
 \Big).
\end{equation}
\end{proposition}
\begin{proof}
As above, we identify the set of radial functions from $\R^{k}\to [-1,1]$ as functions from $(0,\infty)\to[-1,1]$ via the identification $t=\|x\|^2/2$, $x\in\R^k$.  Let $A$ and $C$ be independent gamma random variables with shape parameters $\alpha$ equal to
$2$ and $(k-4)/2$, respectively, and scale $\theta=1$.  Then $A$ has law $\nu_4$, $C$ has law $\nu_{k-4}$, and $A+C$ has law
$\nu_k$ by \eqref{eq:mu3}, so
\[
 J\colon L_2(\nu_k)\longrightarrow
 L_2(\nu_4)\otimes L_2(\nu_{k-4}),
 \qquad
 (Jf)(a,c)\colonequals f(a+c),\qquad\forall\,a,c>0,
\]
is an isometry.  The Ornstein-Uhlenbeck semigroup commutes with $J$:
\begin{equation}
\label{eq:D-intertwining}
 JT_\eta
 =((T_\eta |_{L_2(\nu_4)})\otimes (T_\eta |_{L_2(\nu_{k-4})}))J.
\end{equation}
This follows from \eqref{tel} Laguerre's addition formula
$L_n^{(k/2 -1)}(a+c)
=\sum_{j=0}^{n}L_j^{(1)}(a)L_{n-j}^{(k/2 - 3)}(c)$, rewritten as
\[
J\ell_n^{(k)}
\stackrel{\eqref{ljdef}}{=}\sum_{j=0}^{n}\frac{s_n^{(k)}}{s_j^{(4)}s_{n-j}^{(k-4)}}\cdot \ell_j^{(4)}\ell_{n-j}^{(k-4)}.
\]
Conditioning on $A+C$ reveals that the adjoint $J^*$ is conditional expectation:
\begin{equation}\label{jadj}
J^*(g)(x)= \E[g(A,C)\mid A+C=x],\qquad\forall\,x>0.
\end{equation}

Conditional on $A+C=x$, the ratio $A/x$ has the
$\operatorname{Beta}(2,(k-4)/2)$ distribution.  Hence
\begin{equation}
\label{eq:beta-gamma-conditional}
 \E[A^{-1/2}\mid A+C=x]
 =x^{-1/2}
 \frac{\Gamma(3/2)\Gamma(k/2)}
 {\Gamma(2)\Gamma((k-1)/2)}.
\end{equation}
Consequently
\begin{equation}\label{jsid}
J^*(M_{a^{-1/2}}|_{L_2(\nu_4)}\otimes I)J
=\frac{\Gamma(3/2)\Gamma(k/2)}
 {\Gamma(2)\Gamma((k-1)/2)}M_{x^{-1/2}}.
\end{equation}
Using \eqref{eq:D-intertwining} and \eqref{jsid} we obtain
\begin{equation}
\label{eq:B-intertwining}
 B_{\eta,k}
\colonequals\beta_{\eta,k}T_{\eta}^{1/2}M_{x^{-1/2}}T_{\eta}^{1/2}
 =r_kJ^*\!\left[
 B_{\eta,4}\otimes T_{\eta}|_{L_2(\nu_{k-4})}
 \right]J,
\end{equation}
\[
r_k\colonequals
\frac{\beta_{\eta,k}}{\beta_{\eta,4}}\frac{\Gamma(2)\Gamma((k-1)/2)}{\Gamma(3/2)\Gamma(k/2)}
\stackrel{\eqref{yup}}{=}
\frac{k}{4}\frac{\Gamma(5/2)}{\Gamma(3/2)}\frac{\Gamma((k-1)/2)}{\Gamma((k+1)/2)}
=\frac{3}{4}\frac{k}{k-1}\stackrel{(k\geq5)}{\leq}\frac{15}{16},
\]
\begin{equation}\label{vjeq}
 V_\eta^{(k)}
 \stackrel{\eqref{eq:B-intertwining}\wedge\eqref{eq:V-operator}}{=}J^*\!\left[
 \left(T_{\eta}|_{L_2(\nu_4)}+r_k B_{\eta,4}\right)
 \otimes T_{\eta}|_{L_2(\nu_{k-4})}
 \right]J.
\end{equation}

Let $f\in L_2(\nu_k)$ with $\nu_k$ mean zero.  For any $j\geq0$, let $f_j\in L_2(\nu_4)$ be defined by $Jf=\sum_{j=0}^{\infty}f_j\otimes \ell_j^{(k-4)}$.  Then using $J(1)=1\otimes 1$ and that $J$ is an isometry,
\begin{equation}\label{f0eq}
\langle f_0,1\rangle_{\nu_4}
=\langle f_0\otimes\ell_0^{(k-4)},1\otimes 1\rangle_{\nu_4\otimes \nu_{k-4}}
=\langle Jf,1\otimes 1\rangle_{\nu_4\otimes \nu_{k-4}}
=\langle J^*Jf,1\rangle_{\nu_k}
=\langle f,1\rangle_{\nu_k}=0,
\end{equation}
since $f$ has mean zero.  Also, applying $T_\eta$ to the $\ell_j^{(k-4)}$ terms, we get
\begin{equation}\label{vexp}
\langle f, V_\eta^{(k)}f\rangle
\stackrel{\eqref{vjeq}}{=}
\langle Jf,\,\left[
 \left(T_{\eta}|_{L_2(\nu_4)}+r_k B_{\eta,4}\right)
 \otimes T_{\eta}|_{L_2(\nu_{k-4})}
 \right]Jf\rangle
\stackrel{\eqref{tel}}{=}\sum_{j=0}^{\infty}\eta^{2j}\langle f_j,
 \left(T_{\eta}+r_k B_{\eta,4}\right)f_j\rangle.
\end{equation}
Since $r_k\leq15/16$ and $B_{\eta,4}\geq0$, we have the following operator inequalities on $L_2(\nu_4)$ $\forall$ $j\geq1$:
\[
T_{\eta}+r_k B_{\eta,4}
\leq T_{\eta}+B_{\eta,4}
\stackrel{\eqref{eq:V-operator}}{=}V_{\eta}^{(4)}
,\qquad
\eta^{2j}(T_\eta+r_k B_{\eta,4})
\leq\eta^2\Big(T_\eta+\frac{15}{16}B_{\eta,4}\Big)
\stackrel{\eqref{eq:Z-def}}{=}Z_\eta.
\]
Combining this with \eqref{f0eq} and \eqref{vexp} gives
\[
\langle f, V_\eta^{(k)}f\rangle
\leq \|V_{\eta}^{(4)}\mid_{\mathrm{span}(1)^\perp}\|\,\|f_0\|^2
+\|Z_\eta\|\sum_{j=1}^{\infty}\|f_j\|^2
\leq \max\Big(\|V_{\eta}^{(4)}\mid_{\mathrm{span}(1)^\perp}\|,\|Z_\eta\|\Big)\|f\|^2.
\]
This concludes the proof.
\end{proof}

We now summarize the above findings.

\begin{proposition}[Certified radial bound]
\label{prop:certificateB}
For every $k\geq3$,
\begin{equation}
\label{eq:certificateB}
 \norm{A_{\eta}^{(k)}}
 \leq\norm{V_{\eta}^{(k)}|_{\1^\perp}}
 <
 \begin{cases}
  .987,&k=3,\, 0<\eta\leq.5843\\
  .989,&k\geq4,\,0<\eta\leq3/5.
 \end{cases}
\end{equation}
\end{proposition}

\begin{proof}
The first inequality follows from \eqref{psdineq}.  Then, the case $k=3$ is Proposition \ref{prop:certificate}; the case $k=4$
is Proposition \ref{prop:k4-certificates}; and all cases $k\geq5$ follow
from Propositions \ref{prop:k4-certificates} and
\ref{prop:k4-reduction}.
\end{proof}

\begin{proof}[Proof of Theorem \ref{thm:positive-centered}]
Combine Propositions \ref{prop:radial-criterion} and \ref{prop:certificateB}.  For the corresponding uniqueness statement of Theorem \ref{thm4}, note that the inequality \eqref{stfsineq} is an equality only when $m(R)=0$, i.e. the average value of $f$ on all spheres centered at the origin is zero.  The equality case then follows from \cite[Lemma 5.4]{hwang21} when $\rho>0$, and for $\rho<0$ by Lemma \ref{lem:oddification} below.  In the latter case, equality means that the even part of $f$ vanishes, so the equality case for $\rho>0$ then applies to the $\rho<0$ case.
\end{proof}

\section{Quantum MAX-CUT hardness}

We now pass from the centered positive inequality to the unrestricted
negative inequality.

\begin{lemma}[Gaussian oddification]
\label{lem:oddification}
Let $0<\eta<1$ and let $f\colon \R^n\to\B^k$ be measurable.  Define
\[
  f_{\mathrm e}(x)=\frac{f(x)+f(-x)}2,
  \qquad
  f_{\mathrm o}(x)=\frac{f(x)-f(-x)}2,\qquad\forall\,x\in\R^n.
\]
Then $f_{\mathrm o}$ is odd, $\B^k$-valued, and mean zero, and
\begin{equation}
\label{eq:oddification}
  \Stab_{-\eta}(f)
  =\Stab_\eta(f_{\mathrm e})-\Stab_\eta(f_{\mathrm o})
  \geq-\Stab_\eta(f_{\mathrm o}).
\end{equation}
\end{lemma}

\begin{proof}
The functions $f_{\mathrm e}$ and $f_{\mathrm o}$ are supported on the even
and odd Hermite-Fourier levels, respectively, and these subspaces are orthogonal.  On
a degree-$d$ Hermite level, $T_{-\eta}$ has eigenvalue
$(-\eta)^d$.  This gives the equality in \eqref{eq:oddification}; the
inequality follows from positivity of $T_\eta$.  The remaining assertions
follow directly from the definition of $f_{\mathrm o}$.
\end{proof}

\begin{proof}[Proof of Theorem \ref{thm:negative-borell}]
Write $\rho=-\eta$ so that $\eta>0$.  By Lemma \ref{lem:oddification} and Theorem
\ref{thm:positive-centered},
\[
  \Stab_{\rho}(f)
  =\Stab_{-\eta}(f)
  \geq-\Stab_\eta(f_{\mathrm o})
  \geq-F^*(k,\eta)
  \stackrel{\eqref{eq:Fstar}}{=}F^*(k,-\eta)
  =F^*(k,\rho).
\]
This proves
\eqref{eq:negative-borell}, completing the proof.
\end{proof}

As noted in the introduction, Theorem \ref{thm:negative-borell} immediately implies Theorem \ref{thm5}, since the correlation $\rho_{\rm BOV,3}$ from \eqref{eq:alpha-bov} satisfies \eqref{rhobov3}.

\section{Rank-k MAX-CUT Hardness}

Theorem \ref{thmk} follows from Theorem \ref{thm:negative-borell} for $k\geq4$ once we prove that $\rho_{\rm BOV,k}$ from \eqref{eq:alpha-bov} satisfies $\rho_{\rm BOV,k}\in(-3/5,0)$ for any $k\geq4$.

\begin{proposition}[Location of the worst correlation]
\label{prop:worst-correlation}
For every $k\geq4$, the minimizer $\rho_{\rm BOV,k}$ in
\eqref{eq:alpha-bov} is unique and satisfies
\[
 -\frac35<\rho_{\rm BOV, k}<0.
\]
\end{proposition}

\begin{proof}
Write $\Phi_k(q)\colonequals F^*(k,q)$ for $0\leq q\leq1$.  Since $\Phi_k$
is odd by \eqref{eq:Fstar}, minimizing \eqref{eq:alpha-bov} is equivalent to minimizing
\[
 a_k(q)\colonequals\frac{1+\Phi_k(q)}{1+q},\qquad q\in[0,1].
\]
The sign of $a_k'(q)$ is the sign of
\[
 H_k(q)=(1+q)\Phi_k'(q)-\Phi_k(q)-1.
\]
The power series in \eqref{eq:Fstar} has strictly positive coefficients,
so $H_k'(q)=(1+q)\Phi_k''(q)>0$ for all $q>0$.  It remains to show $H_k(0)<0$ and $H_k(3/5)>0$.  By its definition,
\[
 H_k(0)=g_k-1<0,
 \qquad
 g_k\colonequals\Phi_k'(0)
 \stackrel{\eqref{eq:Fstar}}{=}\frac{2}{k}
 \left(\frac{\Gamma((k+1)/2)}{\Gamma(k/2)}\right)^2
 \stackrel{\eqref{eq:mk-def}}{=}\frac{m_k^2}{k}\stackrel{\eqref{eq:mk-def}}{<}1.
\]

The hypergeometric
series has known power series expansion ${}_2F_1(1/2,1/2,k/2+1,z)=\sum_{j=0}^\infty \frac{(1/2)_j^2}{(k/2+1)_j j!}z^j$.  Since $\Phi_k(q)=g_k q{}_2F_1 (1/2, 1/2, k/2+1,q^2)$, this leads to a power series expansion for $H_k$ of the form
\[
H_k(q)=g_k-1+g_k\frac{3q^2+2q^3}{2(k+2)}
+g_k\frac{9(5q^4+4q^5)}{8(k+2)(k+4)}
+g_k\sum_{j=3}^\infty \frac{(1/2)_j^2}{(k/2+1)_j j!}((2j+1)q^{2j}+2jq^{2j+1}).
\]
Since all of the coefficients are nonnegative, we can lower bound $H_k(q)$ by dropping the $j\geq3$ terms.  Plugging in $q=3/5$ then gives
\begin{equation}
\label{eq:H-lower}
 H_k(3/5)
 \geq
 g_k\left(1+
 \frac{18900k+102573}{25000(k+2)(k+4)}\right)-1.
\end{equation}
The classical Wallis-ratio bound
\[
 \frac{\Gamma(x+1/2)}{\Gamma(x)}
 >\sqrt{x-\frac14},
 \qquad \forall\,x>\frac14,
\]
gives $g_k>1-\frac1{2k}$.  Also, for $k\geq4$,
\[
 \frac{18900k+102573}{25000(k+2)(k+4)}
 >\frac1{2k-1},
\]
since this inequality is equivalent to
$
 12800k^2+36246k-302573>0
$.  Thus the right-hand side of \eqref{eq:H-lower} is positive.  Since
$H_k$ is strictly increasing, it has exactly one zero in $(0,3/5)$.
The derivative of $a_k$ changes there from negative to positive, proving
the claim.
\end{proof}

\begin{proof}[Proof of Theorem \ref{thmk}]
For $k=3$, the minimizing correlation is
$\rho_{\rm BOV, 3}=-0.5842676623\ldots\in(-.5843,0)$.  For $k\geq4$,
Proposition \ref{prop:worst-correlation} places $\rho_{\rm BOV, k}$ in
$(-3/5,0)$.  Apply Theorem \ref{thm:negative-borell} at $\rho_{\rm BOV, k}$ and
then \cite[Theorem~C.4]{hwang21}.  The resulting hardness factor is
\eqref{eq:alpha-bov}.
\end{proof}

\section{Part II. MAX-3-CUT Preliminaries}

Let $\Theta=(\Theta_1,\Theta_2,\Theta_3)$ be the standard simplex
partition of $\R^2$, consisting of three sectors of angle $2\pi/3$ with
vertex at the origin.  For partitions $A=(A_i)_{i=1}^3$ and
$B=(B_i)_{i=1}^3$ of $\R^n$ into three disjoint measurable sets, put
\begin{equation}\label{eq:Q-def}
 Q_\rho^{(n)}(A,B)
 \colonequals\sum_{i=1}^3\int_{\R^n}\1_{A_i}(x)T_\rho\1_{B_i}(x)\gamma_{n}(x)\dd x.
\end{equation}
For measurable $E\subset\R^n$, write $-E=\{-x\colon x\in E\}$, and
$-A\colonequals(-A_i)_{i=1}^3$.  Below we prove

\begin{theorem}[Standard Simplex, Negative Correlation]\label{thm:main}
Let $n\ge2$, let $\Omega=(\Omega_1,\Omega_2,\Omega_3)$ be a measurable partition of
$\R^n$.  Let $\rho\in[-1/2,0)$.  Then
\[
  Q_\rho^{(n)}(\Omega,\Omega)
  \geq Q_\rho^{(2)}(\Theta,\Theta).
\]
\end{theorem}

Put $\sigma\colonequals-\rho$.  Since
$T_{-\sigma}\1_E=T_\sigma\1_{-E}$ for every measurable $E$, negative noise
stability can be written in bilinear form.  For $m\ge2$, let $\mathcal M_m$
be the infimum of $Q_\sigma^{(m)}(A,B)$ over pairs of partitions satisfying
$\gamma_m(A_i)=\gamma_m(B_i)$ for $i=1,2,3$.  Existence of minimizers and
the optimal bilinear dimension-reduction theorem for negative correlation
\cite[Lemma~7.3 and Theorem~7.9]{Heilman2022}, which refines the method of
Heilman--Tarter~\cite{HeilmanTarter2021}, together with the cylinder
embedding, give $\mathcal M_m=\mathcal M_2$.  Thus, for every partition
$\Omega$ of $\R^n$,
\[
 \sum_{i=1}^3\ip{\1_{\Omega_i}}{T_{-\sigma}\1_{\Omega_i}}_{L_2(\gamma_n)}
 =Q_\sigma^{(n)}(\Omega,-\Omega)
 \ge\mathcal M_n=\mathcal M_2.
\]
Consequently, it is enough to prove, for every pair of partitions $A,B$ of
$\R^2$ satisfying $\gamma_2(A_i)=\gamma_2(B_i)$ for all $i$,
\begin{equation}\label{eq:bilinear-target}
 Q_\sigma^{(2)}(A,B)\ge Q_\sigma^{(2)}(\Theta,-\Theta).
\end{equation}

Let $\dd\nu(r)=e^{-r}\dd r$.  For $r\ge0$, on the circle of radius
$\sqrt{2r}$, let $p_i(r)$ and $q_i(r)$ be the normalized angular measures
of $A_i$ and $B_i$.  That is, $p_i(r)=\int_{A_i\cap \sqrt{2r}S^1}\dd x/(2\pi\sqrt{2r})$, for all $r>0$, $1\leq i\leq 3$.  Thus $p(r),q(r)\in\Delta_3$, where
\[
 \Delta_3=\Big\{z\in[0,1]^3:\ \sum_{i=1}^3z_i=1\Big\}.
\]
Since $\gamma_2(A_i)=\gamma_2(B_i)$ for all $1\leq i\leq 3$, we have
\begin{equation}\label{eq:equal-radial-means}
 \int_0^\infty p_i(r)\dd\nu(r)
 =\int_0^\infty q_i(r)\dd\nu(r),\qquad i=1,2,3.
\end{equation}

Let $\mu$ be normalized Haar measure on the circle $S^1$.  Let $(\xi,\eta)$ be a
$\sigma$-correlated pair of standard Gaussian vectors in $\R^2$, and set
$t=|\xi|^2/2$ and $s=|\eta|^2/2$.  Conditional on $t,s$, the joint angular
density relative to $\mu\otimes\mu$ is \cite[Section 2]{Heilman2023}
\begin{equation}\label{kapdef}
\kappa_a(\theta-\phi)\colonequals e^{a\cos(\theta-\phi)}/I_0(a),
\end{equation}
where $I_0$ is the
modified Bessel function of the first kind of order zero (see \eqref{imdef}) and
\begin{equation}\label{eq:a-def}
 d\colonequals1-\sigma^2,\qquad a\colonequals\frac{2\sigma\sqrt{ts}}d.
\end{equation}
Since $\int\kappa_a(\theta-\phi)\dd\mu(\phi)=1$, measurable sets $E,F$ on
the circle satisfy
\[
 \iint_{E\times F}\kappa_a(\theta-\phi)\dd\mu(\theta)\dd\mu(\phi)
 =\mu(E)-\iint_{E\times F^c}\kappa_a(\theta-\phi)
      \dd\mu(\theta)\dd\mu(\phi).
\]
If we wish to maximize the last integral over all measurable $E,F\subset S^1$ with $\mu(E)$ and $\mu(F)$ fixed, then the circular Riesz rearrangement inequality
\cite[Theorem~2]{BaernsteinTaylor1976} implies that 
the maximum occurs for circular arcs with the same center of mass.  Therefore, the first integral is minimized by opposing arcs.  Applying this separately for each $1\leq i\leq 3$ on the left side of \eqref{eq:bilinear-target}, we obtain a lower bound of $Q_\sigma^{(2)}(A,B)$ where $A,B$ are rearranged on a given radius $\sqrt{2r}S^1$ to be opposing arcs.  The three rearranged arcs on
either circle need not be disjoint for the purpose of proving \eqref{eq:bilinear-target}.

We therefore proceed to minimize the spherical noise stability over circular arcs, as in \cite{Heilman2023}.

\section{The joint arc inequality}\label{sec:joint-arcs}

Our first task is proving a strengthened version of \cite[Lemma 3.3]{Heilman2023} for negative correlations.  From now on, set
\begin{equation}\label{eq:b-def}
 b:=\frac13.
\end{equation}
Thus $b$ is the normalized angular length of each of the three sectors in
the standard simplex partition, and $b\1=(1/3,1/3,1/3)$ is the center
of $\Delta_3$.  Identify the circle with
$\mathbb T=\R/(2\pi\mathbb Z)$ and write $\dd\mu(\theta)=\dd\theta/(2\pi)$.
For any $p,q\in[0,1]$, let
\[
 E_p\colonequals\{\theta\in\mathbb T\colon|\theta|\le\pi p\},\qquad
 F_q\colonequals\{\phi\in\mathbb T\colon|\phi-\pi|\le\pi q\},
\]
where both descriptions are understood modulo $2\pi$.  These are arcs of
normalized lengths $p$ and $q$ with opposing centers of mass.  Their conditional
overlap under the angular kernel \eqref{kapdef} is
\begin{equation}\label{eq:K-integral}
 K_a(p,q)
 \colonequals\iint_{E_p\times F_q}
 \frac{e^{a\cos(\theta-\phi)}}{I_0(a)}
 \dd\mu(\theta)\dd\mu(\phi),
 \qquad\forall\,p,q\in[0,1],\,a\in\R.
\end{equation}

We now derive a Fourier formula for $K_a$ and the mixed-derivative
identity used below.  For an integer $m\ge0$, the modified Bessel function
of the first kind has the integral representation
\begin{equation}\label{imdef}
 I_m(a)\colonequals\frac1{2\pi}\int_{-\pi}^{\pi}
 e^{a\cos\vartheta}\cos(m\vartheta)\dd\vartheta,\qquad\forall\,a\in\R.
\end{equation}
It also has the power-series representation
\begin{equation}\label{eq:bessel-power-series}
 I_m(a)=\sum_{\ell\ge0}
 \frac{(a/2)^{m+2\ell}}{\ell!(m+\ell)!},\qquad m\ge0.
\end{equation}

We now derive a generating function of these Bessel functions.  We write
\[
 e^{az/2}e^{a/(2z)}
 =\sum_{r,s\ge0}\frac{(a/2)^{r+s}}{r!s!}z^{r-s},\qquad\forall\,a\in\R,\,z\in\C\setminus\{0\}.
\]
For $m\ge0$, the coefficient of $z^m$ is obtained by setting
$r=s+m$ and is exactly the series for $I_m(a)$ above; the coefficient of
$z^{-m}$ is the same.  With the standard convention
$I_{-m}(a):=I_m(a)$ for $m\ge1$, this proves
\begin{equation}\label{eq:bessel-generating}
 \exp\!\left(\frac a2(z+z^{-1})\right)
 =\sum_{m\in\mathbb Z}I_m(a)z^m,\qquad \forall\,a\in\R,\,z\in\C\setminus\{0\}.
\end{equation}
Taking $z=e^{i\vartheta}$ for some $\vartheta\in\R$ gives
\begin{equation}\label{eq:bessel-fourier}
 \frac{e^{a\cos\vartheta}}{I_0(a)}
 =1+2\sum_{m\ge1}\frac{I_m(a)}{I_0(a)}\cos(m\vartheta).
\end{equation}
For bounded $2\pi$-periodic $f\colon\R\to\R$ and $m\in\Z$ define $\widehat f(m)\colonequals\int_{\mathbb T}f(\theta)e^{-im\theta}
\dd\mu(\theta)$.  The nonzero Fourier coefficients of the two arcs are
\[
 \widehat{\1_{E_p}}(m)=\frac{\sin(\pi mp)}{\pi m},\qquad
 \widehat{\1_{F_q}}(m)=(-1)^m\frac{\sin(\pi mq)}{\pi m}
 \quad(m\ne0),
\]
and their zeroth coefficients are $p$ and $q$, respectively.  The factor $(-1)^m$ comes
precisely from translating the second arc by $\pi$.  Substituting these
coefficients and \eqref{eq:bessel-fourier} into \eqref{eq:K-integral} gives
\begin{equation}\label{eq:K-def}
 K_a(p,q)=pq+2\sum_{m\ge1}(-1)^m\frac{I_m(a)}{I_0(a)}
 \frac{\sin(\pi mp)\sin(\pi mq)}{\pi^2m^2},
  \qquad\forall\,p,q\in[0,1],\,a\in\R.
\end{equation}

For fixed $a\ge0$, the coefficients $I_m(a)$ have inverse factorial decay in $m$ by \eqref{eq:bessel-power-series}.
Thus the series obtained by differentiating \eqref{eq:K-def} once in each
variable converges absolutely and uniformly on $[0,1]^2$.  Writing
$\partial_{pq}\colonequals\partial^2/(\partial p\,\partial q)$, we obtain, for all $a\in\R$ and $p,q\in[0,1]$,
\begin{align*}
 \partial_{pq}K_a(p,q)
 &=1+2\sum_{m\ge1}(-1)^m\frac{I_m(a)}{I_0(a)}
       \cos(\pi mp)\cos(\pi mq)\\
 &=1+\sum_{m\ge1}(-1)^m\frac{I_m(a)}{I_0(a)}
   \Big(\cos\bigl(\pi m(p-q)\bigr)
          +\cos\bigl(\pi m(p+q)\bigr)\Big).
\end{align*}
Replacing $\vartheta$ by $\pi+u$ in \eqref{eq:bessel-fourier} yields
\[
 1+2\sum_{m\ge1}(-1)^m\frac{I_m(a)}{I_0(a)}\cos(mu)
 =\frac{e^{-a\cos u}}{I_0(a)}.
\]
Averaging the two instances of this identity with $u=\pi(p-q)$ and
$u=\pi(p+q)$ proves
\begin{equation}\label{eq:Kpq}
 \partial_{pq}K_a(p,q)=
 \frac{e^{-a\cos(\pi(p-q))}+e^{-a\cos(\pi(p+q))}}{2I_0(a)}.
\end{equation}

Recall $b=1/3$ by \eqref{eq:b-def}.  For any $a\geq0$ define
\begin{equation}\label{kptjdef}
 k(a)\colonequals K_a(b,b),\quad P(a)\colonequals\partial_1K_a(b,b),\quad
 \tau(a)\colonequals4P(a)-3k(a),\quad j(a)\colonequals P(a)-3k(a).
\end{equation}
Here $\partial_1$ denotes differentiation in the first scalar argument.  Formula
\eqref{eq:bessel-power-series} gives $I_0(0)=1$ and $I_m(0)=0$ for every
$m\ge1$.  Hence by \eqref{kptjdef},
\begin{equation}\label{eq:zero-parameters}
 \begin{gathered}
 K_0(p,q)=pq,\qquad k(0)=K_0(b,b)=b^2=\frac19,\qquad
 P(0)=\partial_1K_0(b,b)=b=\frac13,\\
 \tau(0)=4P(0)-3k(0)=\frac43-\frac13=1,
 \qquad
 j(0)=P(0)-3k(0)=\frac13-\frac13=0.
 \end{gathered}
\end{equation}

\begin{lemma}[joint signed arc bound]\label{lem:joint-arc}
For every $a\ge0$ and $p,q\in\Delta_3$, with $b=1/3$ as above, put
$x=p-b\1$ and $y=q-b\1$.  Then
\begin{equation}\label{eq:joint-arc}
 \sum_{i=1}^3K_a(p_i,q_i)-3K_a(b,b)
 \ge \tau(a)\ip{x}{y}+j(a)(\norm{x}^2+\norm{y}^2).
\end{equation}
Moreover, $j(a)>0$ when $a>0$.
\end{lemma}

\begin{proof}
If $a=0$, applying the identities in \eqref{eq:zero-parameters}
coordinatewise gives
\[
 \sum_{i=1}^3K_0(p_i,q_i)-3K_0(b,b)
 =\sum_{i=1}^3p_iq_i-3b^2
 =\ip{p-b\1}{q-b\1}.
\]
The last equality uses $p,q\in\Delta_3$ and $3b=1$.
Since $\tau(0)=1$ and $j(0)=0$, this is exactly equality in
\eqref{eq:joint-arc}.  Henceforth assume $a>0$ and suppress the dependence of
$k,P,\tau,j$ on $a$.  In the scalar argument below, $p,q$ denote arbitrary
numbers in $[0,1]$, rather than the vectors in the lemma statement.  We prove
the stronger scalar inequality
\begin{equation}\label{eq:scalar-arc}
 K_a(p,q)\ge \tau pq-j\{p(1-p)+q(1-q)\},\qquad\forall\,a>0,\,p,q\in[0,1].
\end{equation}
Let 
\begin{equation}\label{radef}
L\colonequals 2b=2/3,\qquad
 r_a(u)\colonequals\frac{e^{-a\cos(\pi u)}}{I_0(a)},\qquad\forall\,u\in\R.
\end{equation}
The function $r_a$ is even and $2$-periodic.
Integrating \eqref{eq:Kpq} from the lower edges of the square gives
\begin{equation}\label{eq:kP-integrals}
 P=\frac12\int_0^L r_a(u)\dd u,\qquad
 k=\frac12\int_0^L (L-u)r_a(u)\dd u.
\end{equation}

We will now prove the following inequalities.
\begin{align}
 j&=\frac12\int_0^L (3u-1)r_a(u)\dd u\ge0,\label{eq:j-positive}\\
 \tau&=\int_0^L (1+3u/2)r_a(u)\dd u\le1,\label{eq:tau-one}\\
 \tau+2j&=\frac92\int_0^L u r_a(u)\dd u\le r_a(L).
 \label{eq:curvature}
\end{align}
Since $r_a$ is increasing on $[0,L]$, the first inequality follows from
$(3u-1)\{r_a(u)-r_a(1/3)\}\ge0$ and
$\int_0^L(3u-1)\dd u=0$.  For the second, \eqref{imdef} and the substitution $\theta\mapsto\pi-\theta$ give
\[
 \int_0^1r_a(v)\dd v
 =\frac1{\pi I_0(a)}\int_0^\pi e^{-a\cos\theta}\dd\theta=1.
\]
Now use the change of variables $v=L+3u^2/4$:
\[
 \int_L^1 r_a(v)\dd v-\frac32\int_0^L u r_a(u)\dd u
 =\frac32\int_0^L u\{r_a(L+3u^2/4)-r_a(u)\}\dd u\ge0.
\]
Indeed, $r_a$ is increasing on $[0,1]$ and
$L+3u^2/4\ge u$ for $0\le u\le L$.
The last inequality in \eqref{eq:curvature} follows since
$(9u/2)\dd u$ is a probability measure on $[0,L]$.  Strict monotonicity
gives $j>0$ for $a>0$.  

Fix $a>0$.  For any $p,q\in[0,1]$ define
\begin{equation}\label{phidef}
 \Phi(p,q)\colonequals K_a(p,q)-\tau pq+j[p(1-p)+q(1-q)].
\end{equation}
Then $\Phi$ is nonnegative on all four boundary edges of the square $[0,1]^2$.  For example,
\[
 \Phi(p,0)=jp(1-p),\qquad
 \Phi(p,1)=(1-\tau)p+jp(1-p).
\]
Here we used $K_a(p,0)=0$ and $K_a(p,1)=p$ by \eqref{eq:K-def}.  The other two
boundary edges follow by symmetry.  Now, if $\Phi$ were negative somewhere on $[0,1]^2$,
its minimum would therefore occur at an interior
critical point.  To analyze such a point, put
\begin{equation}\label{gdef}
 G(\delta)\colonequals\int_0^\delta r_a(u)\dd u,\qquad\forall\,\delta\in[-1,1].
\end{equation}
For $\delta<0$, this is understood as a signed integral.
Since $r_a$ is even, $G$ is odd.  Also, $K_a(p,0)=K_a(0,q)=0$.
We can then integrate \eqref{eq:Kpq} from the corresponding lower
edge of $[0,1]^2$ to get
\begin{align}
 \partial_1K_a(p,q)
 &=\int_0^q\partial_{12}K_a(p,v)\dd v
 =\frac12\int_0^q[r_a(p-v)+r_a(p+v)]\dd v\notag\\
 &=\frac12[G(p+q)-G(p-q)],\label{eq:K-first-p}\\
 \partial_2K_a(p,q)
 &=\int_0^p\partial_{12}K_a(u,q)\dd u=\frac12\int_0^p[r_a(u-q)+r_a(u+q)]\dd u\notag\\
 &=\frac12[G(p+q)+G(p-q)],\qquad\forall\,p,q\in[0,1].
 \label{eq:K-first-q}
\end{align}
Subtracting \eqref{eq:K-first-p} and \eqref{eq:K-first-q} and using the oddness of $G$ gives
\begin{equation}\label{eq:K-first-difference}
 \partial_1K_a(p,q)-\partial_2K_a(p,q)
 =-G(p-q)=G(q-p)\qquad\forall\,p,q\in[0,1],\,a\in\R.
\end{equation}

At an interior critical point of $\Phi$, we have $\nabla\Phi=0$, which by \eqref{phidef} says
\begin{align*}
 0&=\partial_1K_a(p,q)-\tau q+j(1-2p),\\
 0&=\partial_2K_a(p,q)-\tau p+j(1-2q).
\end{align*}
If the critical point $(p,q)$ is off the diagonal, put $\delta\colonequals q-p\ne0$.  Subtracting the
second critical point equation from the first and applying
\eqref{eq:K-first-difference} gives
\begin{equation}\label{gdeq}
 G(\delta)=(\tau-2j)\delta.
\end{equation}

Now, differentiating \eqref{eq:K-first-p} and
\eqref{eq:K-first-q} once more gives
\begin{equation}\label{parteq}
\begin{aligned}
 \partial_{11}K_a(p,q)=\partial_{22}K_a(p,q)
  &=\frac12\{r_a(p+q)-r_a(p-q)\},\\
 \partial_{12}K_a(p,q)
  &=\frac12\{r_a(p+q)+r_a(p-q)\},\qquad\forall\,p,q\in[0,1],\,a\in\R.
\end{aligned}
\end{equation}
Consequently, the unnormalized directional derivative satisfies
\[
 \partial_{(1,-1)}^2\Phi(p,q)
 \colonequals\left.\frac{\dd^2}{\dd t^2}\Phi(p+t,q-t)\right|_{t=0},
\]
\[
 \partial_{(1,-1)}^2\Phi
 =\partial_{11}\Phi+\partial_{22}\Phi-2\partial_{12}\Phi
 \stackrel{\eqref{phidef}\wedge\eqref{parteq}}{=}-2r_a(\delta)+2\tau-4j
 \stackrel{\eqref{gdeq}}{=}-2\left(r_a(\delta)-\frac{G(\delta)}\delta\right).
\]
Since $r_a$ is even and strictly increasing on $[0,1]$ by \eqref{radef}, while $0<|\delta|<1$, we have
\[
 \frac{G(\delta)}\delta
 \stackrel{\eqref{gdef}}{=}\frac1{|\delta|}\int_0^{|\delta|}r_a(u)\dd u
 <r_a(|\delta|)=r_a(\delta).
\]
Thus $\partial_{(1,-1)}^2\Phi<0$, so such a point cannot be a local
minimum.  

It remains to consider the case $\delta=0$, i.e. $p=q$.  Define then
\[
 g(p)\colonequals\Phi(p,p)=K_a(p,p)-\tau p^2+2jp(1-p),\qquad\forall\,p\in[0,1].
\]
By \eqref{eq:K-def}, $K_a(0,0)=0$ and
$K_a(1,1)=1$, while symmetry and the definitions \eqref{kptjdef} give
$K_a(b,b)=k$ and
$\partial_1K_a(b,b)=\partial_2K_a(b,b)=P$.  Therefore
\begin{align*}
 9g(b)&=9k-\tau+4j=0,\\
 g'(b)&=2P-\frac23\tau+\frac23j=0,
\end{align*}
where we used $\tau=4P-3k$ and $j=P-3k$ by \eqref{kptjdef}.  Consequently,
\[
 g(0)=g(b)=g'(b)=0,\qquad g(1)=1-\tau\ge0.
\]
Equations \eqref{eq:K-first-p} and \eqref{eq:K-first-q} also give
\[
 \frac{\dd}{\dd p}K_a(p,p)=G(2p),\qquad
 \frac{\dd^2}{\dd p^2}K_a(p,p)=2r_a(2p).
\]
It follows that
\begin{equation}\label{eq:g-second}
 g''(p)=2[r_a(2p)-\tau-2j].
\end{equation}
Since $p\mapsto r_a(2p)$ increases on $[0,1/2]$ and decreases on $[1/2,1]$,
the set on which \eqref{eq:g-second} is nonnegative is an interval of the form
$[s,1-s]$ for some $0\leq s\leq1/2$.  Inequality \eqref{eq:curvature} gives $s\le b$.  Convexity and $g'(b)=0$ give $g\ge0$ on $[s,1-s]$.  On each
complementary interval in $[0,1]\setminus[s,1-s]$, $g$ is concave, so it lies above the line joining
its nonnegative endpoint values.  Thus \eqref{eq:scalar-arc} holds.

To conclude the proof of \eqref{eq:joint-arc} from \eqref{eq:scalar-arc},
note that $\sum_{i=1}^{3} x_i=\sum_{i=1}^{3} y_i=0$ and hence
\[
 \sum_{i=1}^3p_iq_i=\frac13+\ip{x}{y},\qquad
 \sum_{i=1}^3p_i(1-p_i)=\frac23-\norm{x}^2,\qquad
 \sum_{i=1}^3q_i(1-q_i)=\frac23-\norm{y}^2.
\]
Summing \eqref{eq:scalar-arc} over $i$ gives a constant
$(\tau-4j)/3=3k$, since $\tau-4j=9k$ by \eqref{kptjdef} yielding
\eqref{eq:joint-arc}.
\end{proof}

\subsection{Gain from the first angular mode}

We need a modest strengthening of Lemma~\ref{lem:joint-arc} at the
negative endpoint.  

We first record the homogeneous form of the preceding argument.  Let $r$
be a continuous, nonnegative, even, $2$-periodic function that is nondecreasing on $[0,1]$.  Define $\overline K_r\colon[0,1]^2\to\R$ by
\begin{equation}\label{eq:Kr-mixed}
\overline K_r(p,q)
\colonequals\int_{0}^{p}\int_{0}^{q}
\frac12[r(u-v)+r(u+v)]\dd v \dd u
,\qquad\forall\,p,q\in[0,1].
\end{equation}
Put
\begin{equation}\label{krcon}
 k_r= \overline K_r(b,b),\qquad P_r=\partial_1 \overline K_r(b,b),\qquad
  \tau_r=4 P_r-3 k_r,\qquad  j_r= P_r-3 k_r.
\end{equation}
The proof of Lemma~\ref{lem:joint-arc} gives, without any change other
than normalization,
\begin{equation}\label{eq:homogeneous-arc}
 \sum_{i=1}^3\overline K_r(p_i,q_i)-3 k_r
 \ge  \tau_r\ip{x}{y}+ j_r(\norm{x}^2+\norm{y}^2).
\end{equation}
Indeed, if $M=\int_0^1r(u)\dd u$, then the boundary identity used there is
$\overline K_r(p,1)=Mp$, and the inequality $\tau_r\le1$ is replaced by the
homogeneous inequality $\tau_r\le M$.  All other steps are linear in $r$.
Equivalently, one may apply the normalized lemma to $r/M$ and then multiply
by $M$; the case $M=0$ is immediate.  If $r$ is only nondecreasing rather
than strictly increasing, apply the argument to
$r+\varepsilon r^{(1)}$ and let $\varepsilon\downarrow0$.  All quantities
in \eqref{eq:homogeneous-arc} depend continuously and linearly on $r$.

When $r^{(1)}(u)
=1-\cos(\pi u)$ for all $u\in[0,1]$, we compute from \eqref{krcon} that
\begin{equation}\label{eq:first-mode-def}
 r^{(1)}(u)
 \colonequals1-\cos(\pi u),\,\,
 j_{(1)}\colonequals j_{r^{(1)}}=\frac{9-\sqrt3\,\pi}{4\pi^2},\,\,
 \tau_{(1)}\colonequals\tau_{r^{(1)}}=1+\frac{9-4\sqrt3\,\pi}{4\pi^2},\,\,\forall\,u\in[0,1].
\end{equation}
Integrating \eqref{eq:Kr-mixed} and using the cosine addition formula gives
\begin{equation}\label{eq:first-mode-K}
 \overline K_{(1)}(p,q)\colonequals \overline K_{r^{(1)}}(p,q)
 =pq-\frac{\sin(\pi p)\sin(\pi q)}{\pi^2}.
\end{equation}
Direct substitution at $p=q=b$ gives the constants in
\eqref{eq:first-mode-def}.

The following is a variant of \cite[Lemma~5.4]{Heilman2023}.

\begin{lemma}[simplex trigonometric estimate]\label{lem:simplex-trig}
For any $m=(m_1,m_2,m_3)\in\Delta_3$,
\begin{equation}\label{eq:simplex-trig}
 3+2\sum_{i=1}^3\cos(2\pi m_i)
 \ge6\sum_{i=1}^3(m_i-b)^2.
\end{equation}
\end{lemma}

\begin{proof}
Choose $z\colonequals\min_{i=1,2,3} m_i\in[0,1/3]$ and, after permuting the
coordinates, assume that $m_1=z$ and $m_2\ge m_3$.  Put
$w=m_2-m_3$.  Since $m_2+m_3=1-z$, we have
\begin{equation}\label{mdefs}
 m_1=z,\qquad
 m_2=\frac{1-z+w}{2},\qquad
 m_3=\frac{1-z-w}{2}.
\end{equation}
The condition $m_3\ge z$ shows that $0\le w\le1-3z$.

We now calculate the difference between the two sides of
\eqref{eq:simplex-trig}.  By the cosine addition formula,
\begin{align*}
 3+2\sum_{i=1}^3\cos(2\pi m_i)
 &=3+2\cos(2\pi z)
   +2\cos\bigl(\pi(1-z+w)\bigr)
   +2\cos\bigl(\pi(1-z-w)\bigr)\\
 &=3+2\cos(2\pi z)
   -4\cos(\pi z)\cos(\pi w)\\
 &=1+4\cos^2(\pi z)
   -4\cos(\pi z)\cos(\pi w)
   =1+4\cos(\pi z)[\cos\pi z-\cos\pi w].
\end{align*}
On the other hand, using $\sum_{i=1}^3 m_i=1$ then \eqref{mdefs}
\[
 6\sum_{i=1}^3(m_i-1/3)^2
 =6\Big(-\frac13+\sum_{i=1}^3m_i^2\Big)
 =6z^2+3(1-z)^2+3w^2-2
 =9z^2-6z+1+3w^2.
\]
Consequently, the left side of \eqref{eq:simplex-trig} minus its
right side is
\[
 F_z(w)\colonequals
 3z(2-3z)-3w^2
 -4\cos(\pi z)[\cos(\pi w)-\cos(\pi z)].
\]
It therefore remains to prove $F_z(w)\ge0$ for
$0\le z\le1/3$ and $0\le w\le1-3z$.  We have
\[
 F_z'(0)=0,\qquad
 F_z'''(w)=-4\pi^3\cos(\pi z)\sin(\pi w)\le0,\qquad\forall\,0\leq w\leq 1.
\]
Thus $F_z''$ is decreasing on $[0,1]$, so $F_z'$ can change from positive to negative
at most once on $[0,1-3z]$.  The minimum of $F_z$ is therefore attained at
$w=0$ or $w=1-3z$.  

It remains to show that $F_z(0)\geq0$ and $F_z(1-3z)\geq0$.

Put $t\colonequals3z$.  The condition $F_z(0)\geq0$ is equivalent to
\[
 4c(1-c)\le t(2-t),\qquad c\colonequals\cos(\pi t/3).
\]
Since $t\mapsto\cos(\pi t/3)$ is concave on $[0,1]$, it is lower bounded by its secant line, i.e. $c\ge1-t/2$, which implies
$4y(1-y)\le2t-t^2$ when $0\le y=1-c\le t/2$ by concavity of $y\mapsto y(1-y)$ on $[0,1/2]$.    

When $w=1-t$, the triple angle formula $\cos\pi t=4\cos^3(t\pi/3)-3\cos(t\pi/3)=4c^3 - 3c$ gives
\[
 F_z(1-t)=(4c^2-1)^2-4(1-t)^2.
\]
Again $c\ge1-t/2$ gives
$4c^2-1\ge4(1-t/2)^2-1\ge2(1-t)$ $\forall$ $t\in[0,1]$, proving $F_z(1-t)\geq0$.
\end{proof}

We now prove the following improvement to \eqref{eq:homogeneous-arc} when $r=r^{(1)}$.

\begin{lemma}[first-mode endpoint gain]\label{lem:first-mode-gain}
For any $p,q\in\Delta_3$, put $x\colonequals p-b\1$ and $y\colonequals q-b\1$.  Then
\begin{equation}\label{eq:first-mode-gain}
 \sum_{i=1}^3\overline K_{(1)}(p_i,q_i)-3\overline K_{(1)}(b,b)
 \ge(\tau_{(1)}+2 j_{(1)})\ip{x}{y}
 +\frac32 j_{(1)}(\norm{x}^2+\norm{y}^2).
\end{equation}
\end{lemma}

\begin{proof}
Let
\[
 m\colonequals\frac{p+q}{2},\qquad M\colonequals m-b\1,\qquad
 \delta\colonequals\frac{p-q}{2},
\]
and put
\begin{equation}\label{sdef}
 S\colonequals\norm{x}^2+\norm{y}^2+4\ip{x}{y}
 =\norm{M+\delta}^2 +\norm{M-\delta}^2 + 4\langle M+\delta,M-\delta\rangle
   =6\norm{M}^2-2\norm{\delta}^2.
\end{equation}
Subtract the right side of \eqref{eq:first-mode-gain} from its left side.
Using \eqref{eq:first-mode-K} and
\[
 \sin(\pi(m_i+\delta_i))\sin(\pi(m_i-\delta_i))
 =\sin^2(\pi m_i)-\sin^2(\pi\delta_i),
\]

\[
\sum_{i=1}^3\overline K_{(1)}(p_i,q_i)-3\overline K_{(1)}(b,b)
=\langle x,y\rangle\
+\frac1{\pi^2}
\Big(\frac94-\sum_{i=1}^3\sin^2(\pi m_i)
+\sum_{i=1}^3\sin^2(\pi\delta_i)\Big).
\]

Subtracting the right-hand side of \eqref{eq:first-mode-gain} from this, we get a deficit
\[
\Delta
=
\frac1{\pi^2}
\Big(\frac94-\sum_{i=1}^3\sin^2(\pi m_i)
+\sum_{i=1}^3\sin^2(\pi\delta_i)\Big)
+\bigl(1- \tau_{(1)}-2 j_{(1)}\bigr)\langle x,y\rangle
-\frac32 j_{(1)}\bigl(\|x\|^2+\|y\|^2\bigr).
\]
Letting $\mathcal D\colonequals 4\pi^2 \Delta$, and using \eqref{eq:first-mode-def} with $\|x\|^2+\|y\|^2=2\|M\|^2 + 2\|\delta\|^2$ and $\langle x,y\rangle=\|M\|^2-\|\delta\|^2$,

\begin{equation} \label{eq:first-mode-deficit}
\mathcal D
=3+2\sum_{i=1}^3\cos(2\pi m_i)
 +4\sum_{i=1}^3\sin^2(\pi\delta_i)
+9(\sqrt3\,\pi-6)\norm{M}^2
 -3\sqrt3\,\pi\norm{\delta}^2.
\end{equation}
Lemma~\ref{lem:simplex-trig} and the concave secant line inequality $\sin(\pi|u|)\ge2|u|$ for $|u|\le1/2$ imply
\[
 \mathcal D\ge
 (3\sqrt3\,\pi-16)(3\norm{M}^2-\norm{\delta}^2)
 \stackrel{\eqref{sdef}}{=}\frac{3\sqrt3\,\pi-16}{2}S.
\]
Since $3\sqrt3\,\pi-16>0$, this proves the assertion \eqref{eq:first-mode-gain} when $S\ge0$.
When $S<0$, the right side of \eqref{eq:first-mode-gain} minus the
right side of the homogeneous estimate \eqref{eq:homogeneous-arc} is
$j_{(1)}S/2<0$ by \eqref{eq:first-mode-def}.  Hence \eqref{eq:homogeneous-arc} proves the case $S<0$ of \eqref{eq:first-mode-gain}.  So, in all cases, \eqref{eq:first-mode-gain} holds.
\end{proof}

Define
\begin{equation}\label{eq:theta-def}
 \theta(a)
 \colonequals\frac{j_{(1)}}{2}\frac{ae^{-a}}{I_0(a)},\qquad\forall\,a\geq0.
\end{equation}
Since
\[
 r_a(u)
 \stackrel{\eqref{radef}}{=}\frac{e^{-a\cos(\pi u)}}{I_0(a)}
 =r_a(0)e^{a(1-\cos(\pi u))},\qquad\forall\,u\in\R,
\]
we have the exact decomposition
\begin{equation}\label{eq:angular-decomposition}
 r_a(u)=r_a(0)[1+a r^{(1)}(u)]+r_{\mathrm{rem}}(u),\qquad\forall\,u\in\R,
\end{equation}
where
\[
 r_{\mathrm{rem}}(u)\colonequals r_a(0)
 \left(e^{a(1-\cos(\pi u))}-1-a(1-\cos(\pi u))\right),\qquad\forall\,u\in\R.
\]
The remainder is nonnegative and nondecreasing on $[0,1]$, since
\[
 r_{\mathrm{rem}}'(u)=r_a(0)a\pi\sin(\pi u)
 \left(e^{a(1-\cos(\pi u))}-1\right)\ge0,\qquad\forall\,a\geq0,\,u\in[0,1].
\]

Observe that \eqref{eq:Kr-mixed} and all quantities defined in \eqref{krcon} are linear in $r$.  We then apply \eqref{eq:homogeneous-arc} to three different choices of $r$, corresponding to the three functions on the right side of \eqref{eq:angular-decomposition}: (i) the constant function $r=r_a(0)$, (ii) $r\colonequals r_{a}(0)ar^{(1)}$, and (iii) $r=r_{\rm rem}$.  By linearity, and recalling $K_a=\overline K_{r_a}$, the quantity $\sum_{i=1}^3K_a(p_i,q_i)-3K_a(b,b)$ is then lower bounded by a sum of three terms corresponding to (i), (ii) and (iii).  In cases (i) and (iii) we apply \eqref{eq:homogeneous-arc}, and in case (ii) we apply the improved bound \eqref{eq:first-mode-gain}.  Summing up all three contributions then gives
\begin{corollary}
For any $a\geq0$ and $p,q\in\Delta_3$, put $x=p-b\1$ and $y=q-b\1$.  Then
\begin{equation}\label{eq:endpoint-joint-arc}
 \sum_{i=1}^3K_a(p_i,q_i)-3K_a(b,b)
 \ge(\tau(a)+4\theta(a))\ip{x}{y}
 +(j(a)+\theta(a))(\norm{x}^2+\norm{y}^2).
\end{equation}
\end{corollary}

Note that this is an improvement to \eqref{eq:joint-arc}.  The first gained term comes from the right side of \eqref{eq:first-mode-gain} where $r_{a}(0)a2j_{(1)}=4\theta(a)$ via \eqref{eq:theta-def}, and the second gained term comes from $r_a(0)aj_{(1)}/2=\theta(a)$.

\section{The radial operator}

We first derive the radial law used in this section.  Recall that
$(\xi,\eta)$ is a $\sigma$-correlated pair of standard Gaussian vectors in
$\R^2$ and
\[
 t\colonequals\frac{\|\xi\|^2}{2},\qquad s\colonequals\frac{\|\eta\|^2}{2}.
\]
Thus the corresponding Euclidean radii are $\sqrt{2t}$ and $\sqrt{2s}$.
Equivalently, if $Z$ is an independent standard Gaussian vector, then
$\eta=\sigma\xi+\sqrt d\,Z$, where $d=1-\sigma^2$.  Since
$\|\xi\|^2$ and $\|\eta\|^2$ are each $\chi^2_2$-distributed, both $t$ and $s$
have the exponential probability law
\[
 \dd\nu(r)=e^{-r}\dd r,\qquad r\ge0.
\]
\[
 \E t=\E s=1,\qquad \operatorname{Var}(t)=\operatorname{Var}(s)=1,
 \qquad \operatorname{Cov}(t,s)=\sigma^2.
\]

The joint Gaussian density of $(\xi,\eta)$ is
\[
 \frac1{(2\pi)^2d}
 \exp\left(-\frac{\|\xi\|^2+\|\eta\|^2-2\sigma\langle\xi,\eta\rangle}{2d}\right).
\]
Suppose we parametrize $\xi,\eta$ as
\[
 \xi=\sqrt{2t}(\cos\theta,\sin\theta),\qquad
 \eta=\sqrt{2s}(\cos\phi,\sin\phi),\qquad\theta,\phi\in[0,2\pi).
\]
Observe $\langle\xi,\eta\rangle=2\sqrt{ts}\cos(\theta-\phi)$.  So the joint density of
$(t,s,\theta,\phi)$, relative to
$\dd t\dd s\dd\mu(\theta)\dd\mu(\phi)$ is
\begin{equation}\label{eq:polar-joint-density}
 \frac1d\exp\left(-\frac{t+s}{d}
       +\frac{2\sigma\sqrt{ts}}d\cos(\theta-\phi)\right).
\end{equation}
Put
\begin{equation}\label{adef}
 a=a(t,s)
 \colonequals\frac{2\sigma\sqrt{ts}}d.
\end{equation}
Since
$I_0(a)=\int_{\mathbb T} e^{a\cos u}\dd\mu(u)$ by \eqref{imdef}, integrating
\eqref{eq:polar-joint-density} over the two angular variables gives the
ordinary joint density of $(t,s)$ relative to $\dd t\,\dd s$:
\begin{equation}\label{eq:radial-lebesgue-density}
 f_{t,s}(t,s)=\frac1d e^{-(t+s)/d}
 I_0\left(\frac{2\sigma\sqrt{ts}}d\right),\qquad t,s\ge0.
\end{equation}
On the other hand, $\nu\otimes\nu$ has Lebesgue density
$e^{-(t+s)}$.  Therefore the Radon--Nikodym derivative of the joint radial
law with respect to $\nu\otimes\nu$ is
\begin{equation}\label{eq:R0-def}
 R_0(t,s)
 =\frac{f_{t,s}(t,s)}{e^{-(t+s)}}
 =\frac1d e^{-\sigma^2(t+s)/d}
 I_0\left(\frac{2\sigma\sqrt{ts}}d\right),\qquad\forall\,t,s\geq0.
\end{equation}
Thus, for every measurable $D\subseteq[0,\infty)^2$,
\[
 \Pr[(t,s)\in D]=\iint_D R_0(t,s)\dd\nu(t)\dd\nu(s).
\]
In particular, since each marginal is $\nu$,
\begin{equation}\label{r0int}
 \int_0^\infty R_0(t,s)\dd\nu(s)=1
 \quad\text{for $\nu$-almost every $t$},
\end{equation}
and likewise with $t$ and $s$ interchanged.
Equivalently, the complete radial-angular factorization is
\[
 \dd\mathbb P
 =R_0(t,s)\dd\nu(t)\dd\nu(s)\,
   \frac{e^{a\cos(\theta-\phi)}}{I_0(a)}
   \dd\mu(\theta)\dd\mu(\phi).
\]

We now apply \eqref{eq:endpoint-joint-arc} at each pair $(t,s)$ of radii.  Put 
\begin{align}
 \mathsf K_0(t,s)&\colonequals R_0(t,s)\tau(a)
 \stackrel{\eqref{eq:tau-one}}{=}\frac{e^{-\sigma^2(t+s)/d}}{2d}
 \int_0^{2/3}(2+3v)e^{-a\cos(\pi v)}\dd v,\label{eq:radial-K}\\
 J_0(t)&\colonequals \int_0^\infty R_0(t,s)j(a)\dd\nu(s)
 \stackrel{\eqref{eq:j-positive}}{=}\int_{0}^{\infty}\frac{e^{-\sigma^2(t+s)/d}}{2d}\int_0^L (3u-1)e^{-a\cos(\pi u)}\dd u \dd\nu(s),\label{eq:J-def}\\
 \mathsf B(t,s)&\colonequals R_0(t,s)\theta(a)
 \stackrel{\eqref{eq:theta-def}\wedge\eqref{adef}}{=}\frac{j_{(1)}}{2d}a\,e^{-\sigma^2(t+s)/d-a},\label{bdef}\\
 B_{(1)}(t)&\colonequals \int_0^\infty\mathsf B(t,s)\dd\nu(s)
 \stackrel{\eqref{adef}}{=}\sigma\sqrt{t}\frac{j_{(1)}}{d^2} \int_{0}^{\infty}\sqrt{s} e^{-2\sigma\sqrt{ts}/d}e^{-\sigma^2(t+s)/d - s}\dd s,\label{eq:B1-def}\\
 \mathsf K_*(t,s)&\colonequals \mathsf K_0(t,s)+4\mathsf B(t,s),
 \qquad J_*(t)\colonequals J_0(t)+B_{(1)}(t).\label{eq:star-kernels}
\end{align}
If $\mathcal K_*$ is the
integral operator with kernel $\mathsf K_*$ and $M_{J_*}$ is multiplication
by $J_*$, then integration of \eqref{eq:endpoint-joint-arc} in $t,s$ gives the
quadratic expression
\begin{equation}\label{eq:integrated-star-form}
 \sum_{i=1}^3\Big(
 \ip{x_i}{\mathcal K_*y_i}
 +\ip{x_i}{M_{J_*}x_i}+\ip{y_i}{M_{J_*}y_i}\Big),
\end{equation}
where $x_i\colonequals p_i-b$, $y_i\colonequals q_i-b$, and $p_i,q_i\colon(0,\infty)\to[0,1]$ are defined after \eqref{eq:bilinear-target}, for all $i=1,2,3$, and here $\langle\cdot,\cdot\rangle$ denotes the $L_2$ inner product on $(0,\infty)$ with respect to the exponential measure $\nu$.

\subsection{Operator norm inequality}

Our first task is to bound the first term in \eqref{eq:integrated-star-form} in Lemma \ref{speclem} below.

We record an elementary weighted Hilbert--Schmidt estimate used twice
below.  If $J>0$ almost everywhere and $\mathsf L\colon(0,\infty)^2\to\R$ is a symmetric bounded kernel,
put
\[
 h(\mathsf L,J)
\colonequals\Big(\int_{0}^{\infty}\int_{0}^{\infty}\frac{\mathsf L(t,s)^2}{J(t)J(s)}\dd\nu(t)\dd\nu(s)\Big)^{1/2}.
\]
Then every bounded real $z\colon(0,\infty)\to\R$ satisfies
\begin{equation}\label{eq:weighted-HS}
 |\ip{z}{\mathcal Lz}|
 \le h(\mathsf L,J)\ip{z}{M_Jz}.
\end{equation}
Indeed, with $f\colonequals\sqrt Jz$ and
$\mathsf H(s,t)\colonequals\mathsf L(s,t)/\sqrt{J(t)J(s)}$, the Cauchy--Schwarz inequality on the space $(0,\infty)^2$ gives
\begin{equation}\label{hsineq}
 |\ip{z}{\mathcal Lz}|=|\ip{f}{\mathcal Hf}|
 \le\norm{\mathsf H}_{L_2(\nu\otimes\nu)}\norm{f}_{L_2(\nu)}^2.
\end{equation}

For the endpoint certificate, define
\begin{align}
 u_*&\colonequals\mathcal K_*\1,\qquad C_*\colonequals\ip{\1}{u_*},\label{usdef}\\
 \mathsf R_*(t,s)&\colonequals\mathsf K_*(t,s)-\frac{u_*(t)u_*(s)}{C_*},\qquad\forall\,t,s>0,
 \label{eq:Rstar-def}\\
 \mathsf E_*(t,s)&=\mathsf K_*(t,s)-u_*(s)-u_*(t)+C_*\qquad\forall\,t,s>0.\label{eq:Estar-def}
\end{align}
Then $R_{*}$ has zero row and column integrals, and
if $\int_{0}^{\infty} z\,d\nu=0$, 
$
  \langle z,E_{*}z\rangle
  =
  \langle z,K_{*}z\rangle.
$
Set
\begin{align}
 h_{\rm lower}
 =h_{\rm lower}(\sigma)&\colonequals\Big(\int_{0}^{\infty}\int_{0}^{\infty}\frac{\mathsf R_*(t,s)^2}{J_*(t)J_*(s)}
             \dd\nu(t)\dd\nu(s)\Big)^{1/2},\label{eq:hall-def}\\
 h_{\rm upper}
 =h_{\rm upper}(\sigma)&\colonequals\Big(\int_{0}^{\infty}\int_{0}^{\infty}\frac{\mathsf E_*(t,s)^2}{J_*(t)J_*(s)}
             \dd\nu(t)\dd\nu(s)\Big)^{1/2}.\label{eq:hmz-def}
\end{align}

\begin{lemma}[Spectral bounds]\label{speclem}
Let $z\colon(0,\infty)\to\R$ be bounded.  Then
\begin{equation}\label{eq:endpoint-lower-operator}
 \ip{z}{\mathcal K_*z}\ge-h_{\rm lower}\ip{z}{M_{J_*}z}.
\end{equation}

\begin{equation}\label{eq:endpoint-upper-operator}
 \ip{z}{\mathcal K_*z}\le h_{\rm upper}\ip{z}{M_{J_*}z}
 \qquad\text{when $\int_{0}^{\infty} z(s)\dd\nu(s)=0$}.
\end{equation}
\end{lemma}
\begin{proof}
We first identify the rank-one part.  By definition,
\[
 u_*(t)=\int_0^\infty\mathsf K_*(t,s)\dd\nu(s),\qquad
 C_*=\int_0^\infty u_*(t)\dd\nu(t)
   =\int_0^\infty\int_0^\infty\mathsf K_*(t,s)\dd\nu(t)\dd\nu(s)>0.
\]
The residual $R_*$ has zero row and column integrals:
\[
 \int_0^\infty\mathsf R_*(t,s)\dd\nu(s)
 =u_*(t)-\frac{u_*(t)}{C_*}\int_0^\infty u_*(s)\dd\nu(s)=0,
\]
and similarly in the other variable.  With the convention
$(u_*\otimes u_*)z=u_*\,\ip{u_*}{z}$, we have
\[
 \mathcal K_*=\mathcal R_*+\frac{u_*\otimes u_*}{C_*},
 \qquad
 \ip{z}{\mathcal K_*z}
 =\ip{z}{\mathcal R_*z}+\frac{|\ip{z}{u_*}|^2}{C_*}.
\]
The normalization here is $C_*=\ip{\1}{u_*}$, rather than
$\norm{u}_{L_2(\nu)}^2$: it is chosen so that $\mathcal R_*\1=0$ and is not
asserted to make $(u_*\otimes u_*)/C_*$ an orthogonal projection.
The rank-one term is nonnegative, so its omission together with
\eqref{eq:weighted-HS} and \eqref{hsineq} with $J=J_*$ prove \eqref{eq:endpoint-lower-operator}.

If $\int_{0}^{\infty} z(s)\dd\nu(s)=0$, then the $u_*,C_*$ terms in \eqref{eq:Estar-def} vanish by Fubini's Theorem:
\[
\int_{0}^{\infty}\int_{0}^{\infty}[-u_*(t)-u_*(s)+C_*]z(t)z(s)\dd\nu(t)\dd\nu(s)=0.
\]
Consequently, \eqref{eq:endpoint-upper-operator} holds by \eqref{eq:weighted-HS}, \eqref{hsineq} and \eqref{eq:Estar-def}.
\end{proof}

\subsection{Eigenvalue Bounds}

We now bound the quantities \eqref{eq:hall-def} and \eqref{eq:hmz-def} from Lemma \ref{speclem}.

We next give the formulas and directed bounds used to certify
$h_{\rm upper}<2$ and $h_{\rm lower}<2$.
Recall that $0<\sigma\leq1/2$, and
\begin{equation}\label{eq:endpoint-parameters}
 d=d_{\sigma}\colonequals1-\sigma^2
 \qquad
 \kappa =\kappa_{\sigma}\colonequals\frac{\sigma^2}{d},
 \qquad
 q=q_{\sigma,t}\colonequals\frac{\sigma\sqrt t}{\sqrt{d}}
\end{equation}
For the uncorrected radial terms, put
\begin{equation}\label{bvdef}
 b_v(t)
 \colonequals q\cdot\cos(\pi v),\qquad\forall\,v\in[0,1],\,t\geq0.
\end{equation}
\begin{equation}\label{eq:F-def}
 F_v(t)\colonequals
 e^{-\kappa t}
 \left[1-\sqrt\pi b_v(t)e^{b_v(t)^2}
                  \operatorname{erfc}(b_v(t))\right],\qquad\forall\,t\geq0.
\end{equation}
Here
$
 \operatorname{erfc}(z)\colonequals\frac2{\sqrt\pi}
       \int_z^\infty e^{-u^2}\dd u$, $\forall\,z\in\R
$
is the complementary error function.
The same Gaussian integration as in the previous section gives 
\begin{lemma}\label{ujlem}
$u_*=u_0+4B_{(1)}$ by \eqref{usdef} and $J_*=J_0+B_{(1)}$ by \eqref{eq:star-kernels} where
\begin{equation}\label{eq:base-uJ}
 u_0(t)\colonequals\mathcal K_0\mathbf{ 1}=\frac12\int_0^{2/3}(2+3v)F_v(t)\dd v,\qquad
 J_0(t)\stackrel{\eqref{eq:J-def}}{=}\frac12\int_0^{2/3}(3v-1)F_v(t)\dd v,\qquad\forall\,t\geq0.
\end{equation}
\end{lemma}
\begin{proof}
Using $d=1-\sigma^2$ and $a$ from \eqref{adef}, then substituting $s=d\cdot y^2$, we get
\begin{equation}\label{ivc}
\begin{aligned}
\mathcal I_v(t)
&\colonequals
\int_0^\infty
e^{-\sigma^2s/d}
e^{-a\cos(\pi v)}\dd \nu(s)
\stackrel{\eqref{adef}}{=}\int_0^\infty
\exp\left(
-\frac{s}{d}
-\frac{2\sigma\sqrt{ts}}d\cos(\pi v)
\right)\dd s\\
&
=d\int_0^\infty
\exp\left(
-y^2
-2\sigma\sqrt{t}y\cos(\pi v)
\right)2y\dd y
=2d\int_0^\infty y e^{-y^{2} - 2b_v(t)y}\dd y\\
&=2d\int_0^\infty y e^{-(y+b_v(t))^2 + b_v(t)^2}\dd y
=2de^{b_v(t)^2}\int_{b_v(t)}^\infty (y-b_v(t)) e^{-y^2}\dd y\\
&=de^{b_v(t)^2}[e^{-b_v(t)^2}-b_v(t)\sqrt{\pi}\mathrm{erfc}(b_v(t))]
=d[1-\sqrt{\pi}b_v(t)e^{b_v(t)^2}\mathrm{erfc}(b_v(t))]
=de^{\kappa t}F_v(t).
\end{aligned}
\end{equation}

\begin{align*}
u_0(t)
&=\int_0^\infty\mathsf K_0(t,s)\dd\nu(s)
\stackrel{\eqref{eq:radial-K}}{=}
\frac{e^{-\sigma^2t/d}}{2d}
\int_0^{2/3}(2+3v)
\left[
\int_0^\infty
e^{-\sigma^2s/d}
e^{-a\cos(\pi v)}\dd\nu(s)
\right]dv\\
&\stackrel{\eqref{ivc}}{=}\frac{e^{-\sigma^2t/d}}{2d}
\int_0^{2/3}(2+3v)
de^{\kappa t}F_v(t)\dd v
\stackrel{\eqref{eq:endpoint-parameters}}{=}\frac12\int_0^{2/3}(2+3v)F_v(t)\dd v.
\end{align*}

Likewise, applying the same calculation to $J_0$,
\[
J_0(t)
\stackrel{\eqref{eq:J-def}}{=}
\int_{0}^{\infty}\frac{e^{-\sigma^2(t+s)/d}}{2d}\int_0^L (3u-1)e^{-a\cos\pi u}\dd u \dd\nu(s)
=\frac12\int_0^{2/3}(3v-1)F_v(t)\dd v.
\]

\end{proof}

From \eqref{eq:base-uJ} and \eqref{eq:F-def}, $F_v(x^2)=1 - \sqrt{\pi}(\sigma/\sqrt{d})x\cos(\pi v)+O(x^2)$ uniformly in $v$, so
\begin{equation}\label{j0as}
J_0(x^2)=\frac{1}{2}\int_{0}^{2/3}(3v-1)(-\sqrt{\pi}(\sigma/\sqrt{d})x\cos\pi v + O(x^2))
\stackrel{\eqref{eq:first-mode-def}}{=}x\sqrt{\pi}\frac{\sigma}{\sqrt{d}}j_{(1)}+O(x^2).
\end{equation}

We now derive a similar asymptotic expansion for $B_{(1)}$.  

If $t=x^2$ and
$q=\sigma x/\sqrt{d}$, direct integration in the second radial variable gives
\begin{equation}\label{eq:B1-closed}
 B_{(1)}(x^2)\stackrel{\eqref{eq:B1-def}}{=}e^{-\kappa x^2}\frac{j_{(1)}q}{2}
 \left[\sqrt\pi(1+2q^2)e^{q^2}\operatorname{erfc}(q)-2q\right].
\end{equation}
\begin{proof}
Using the substitution $y=\sqrt{1+\sigma^2/d}\sqrt{s}=\sqrt{s/d}$ by \eqref{eq:endpoint-parameters},
\begin{flalign*}
B_{(1)}(t)
&=\sigma\sqrt{t}\frac{j_{(1)}}{d^2} \int_{0}^{\infty}\sqrt{s} e^{-2\sigma\sqrt{ts}/d}e^{-\sigma^2(t+s)/d - s}\dd s\\
&=\sqrt{t}j_{(1)}\frac{\sigma}{d^2} e^{-\frac{\sigma^2}{d}t}\,2d^{3/2}\int_{0}^{\infty} y^2 e^{-2y\sigma\sqrt{t}/\sqrt{d}}e^{-y^2}\dd y\\
&=\sqrt{t}j_{(1)}\frac{2\sigma}{\sqrt{d}}\int_{0}^{\infty} y^2 e^{-(y + \sigma\sqrt{t/d})^2}\dd y
=\sqrt{t}j_{(1)}\frac{2\sigma}{\sqrt{d}}
\int_{\sigma\sqrt{t/d}}^{\infty} (y-\sigma \sqrt{t/d})^2 e^{-y^2}\dd y\\
&=2j_{(1)}q\int_{\sigma\sqrt{t/d}}^{\infty} \Big(y^2 - 2y\sigma\sqrt{t/d}+\sigma^2 t/d\Big) e^{-y^2}\dd y\\
&=2j_{(1)}q
\Big(\frac{\sqrt{\pi}}{4}\mathrm{erfc}(q) + (1/2)qe^{-q^2} - qe^{-q^2}+q^2\frac{\sqrt{\pi}}{2}\mathrm{erfc}(q)\Big)\\
\end{flalign*}

\end{proof}
So, $\lim_{x\to0^+}B_{(1)}(x^2)/x\stackrel{\eqref{eq:B1-closed}}{=}\frac{j_{(1)}\sigma}{2\sqrt{d}}\sqrt{\pi}$, $\lim_{x\to0^+}J_0(x^2)/x\stackrel{\eqref{j0as}}{=}\frac{j_{(1)}\sigma}{\sqrt{d}}\sqrt{\pi}$, and by Lemma \ref{ujlem},
$$
\lim_{x\to0^+}J_*(x^2)/x=j_{(1)}\frac{3}{2}\sqrt{\pi}\frac{\sigma}{\sqrt{d}}
\stackrel{\eqref{eq:first-mode-def}}{>}0.
$$
Then, for any $x>0$,
\begin{equation}\label{eq:Jstar-origin}
 \widetilde J_*(x)\colonequals\frac{J_*(x^2)}x,\qquad
 \widetilde J_*(0)=j_{(1)}\frac{3}{2}\sqrt{\pi}\frac{\sigma}{\sqrt{d}}
 \stackrel{\eqref{eq:first-mode-def}}{>}0.
\end{equation}

Consequently, $\widetilde J_*$ is continuous and positive near zero ($J_*(x)>0$ $\forall$ $x>0$ by \eqref{eq:Jstar-lower} below).  The change of variables $t=x^2$, $s=y^2$ therefore transforms both
\eqref{eq:hall-def} and \eqref{eq:hmz-def} into
\begin{equation}\label{eq:endpoint-HS-xy}
 \Big(4\int_0^\infty\int_0^\infty e^{-x^2-y^2}
 \frac{[\mathsf L(x^2,y^2)]^2}
      {\widetilde J_*(x)\widetilde J_*(y)}\dd x\dd y\Big)^{1/2},
\end{equation}
with $\mathsf L=\mathsf R_*$ or $\mathsf E_*$.  We can now state the eigenvalue bounds for \eqref{eq:hall-def} and \eqref{eq:hmz-def}
using \eqref{eq:endpoint-HS-xy}.

\begin{lemma}[Uniform Eigenvalue Bound]\label{lem:certificate}
Uniformly for all $\sigma\in[1/43,1/2]$, we have
\begin{equation}\label{eq:cert-bounds}
 h_{\rm upper}(\sigma)<1.967<2,\qquad h_{\rm lower}(\sigma)<0.824<2.
\end{equation}
\end{lemma}

\begin{proof}

The accompanying program \path{plur_cert.py} uses Arb balls for
transcendental evaluations and outward-rounded binary64 interval arithmetic
for the remaining operations.  It subdivides
$[1/43,1/2]$ into $640$ nonuniform $\sigma$-intervals and $[0,5]$ into
$1000$ intervals in each radial variable.  Here the nonuniform intervals $[\sigma_i,\sigma_{i+1}]$ use $\sigma_i=(1/43)+(1/2 - 1/43)(i/640)^p$ for each $0\leq i\leq 640$, where $p=1.3$.  The auxiliary tables use $64$
second-order interval midpoint panels in the angular variable, together with
$4096$, $8192$, and $65536$ nonuniform intervals (with $p=2$) for $q$, $a$, and the
exponential factor, respectively.

On every box $I_i\times I_j$ in the integrand domain and every $\sigma$-interval, the program
encloses $\widetilde J_*$ away from zero and encloses
$\mathsf R_*^2$ and $\mathsf E_*^2$ from above.  It therefore bounds the
box's contribution to \eqref{eq:endpoint-HS-xy} by
\[
 4|I_i||I_j|
 \frac{\sup_{x\in I_i}e^{-x^2}\sup_{y\in I_j}e^{-y^2}
       \sup_{(x,y)\in I_i\times I_j}\mathsf L(x^2,y^2)^2}
      {\inf_{x\in I_i}\widetilde J_*(x)
       \inf_{y\in I_j}\widetilde J_*(y)},
\]
where $\mathsf L$ is either $\mathsf R_*$ or $\mathsf E_*$, and we additionally take a supremum over $\sigma$ values in the given $\sigma$-interval.  The identity
\[
 C_*-1
 \stackrel{\eqref{usdef}}{=}\int_0^\infty 2xe^{-x^2}\bigl(u_*(x^2)-1\bigr)\dd x
\]
is used instead of integrating $C_*$ directly; this preserves the exact
cancellation at $\sigma=0$ and makes the interval enclosure uniform down to
$\sigma=1/43$.  The enclosure for $C_*$ occurs at lines 599, 618 and 621 of the code.

For clarity, we also record the analytic tail bounds.  From the $\ell=m=0$ term in \eqref{eq:bessel-power-series}, we have $I_0(a)\geq1$ for all $a\in\R$.  Then $\theta(a)\stackrel{\eqref{eq:theta-def}}{\le} j_{(1)}/(2e)\stackrel{\eqref{eq:first-mode-def}}{<}1/50$ and $0\leq u_{0}(t)\leq\int_{0}^{\infty}R_{0}(t,s)\tau(a)d\nu(s)\leq 1$ by \eqref{eq:tau-one} and \eqref{r0int}.  Also, $0\leq j\leq 1/2$ by \eqref{eq:j-positive} and $\int_{0}^{1}r_a(u)du=1$ with $r_{a}\geq0$ by \eqref{radef}, so $0\leq J_0\leq1/2$ by \eqref{eq:J-def}.  Then $0\leq B_{(1)}(t)\leq1/50$ by \eqref{eq:B1-def}.  Lemma \ref{ujlem} then implies
\begin{equation}\label{a0bd}
 0\le u_*(t)\le1+\frac{4}{50}=\frac{27}{25},\qquad
 0\le J_*(t)\le\frac{1}{2}+\frac{1}{50}=\frac{13}{25},\qquad
 0\leq C_*\stackrel{\eqref{usdef}}{\leq}\frac{27}{25}.
 \end{equation}
The $u_*$ estimate is used in lines 618 and 716 in the code.

We include the integration-by-parts proof of the lower bound needed in the
tail estimate.  Define
\begin{equation}\label{psidef}
 \psi(b)=1-\sqrt\pi be^{b^2}\operatorname{erfc}(b)
 \stackrel{\eqref{ivc}}{=}\int_0^\infty2w e^{-w^2-2bw}\dd w,\qquad\forall\,b\in\R.
\end{equation}
Thus $F_v(x^2)=e^{-\kappa x^2}\psi(b_v(x^2))$ by \eqref{ivc}.  Since
$3v-1=-\frac{\dd}{\dd v}[v(1-3v/2)]$ and $v(1-3v/2)$
vanishes at $v=0$ and $v=2/3$, integration by parts in \eqref{eq:base-uJ} gives
\[
 J_0(x^2)=\frac12\int_0^{2/3}v(1-3v/2)
                    \partial_vF_v(x^2)\dd v,\qquad\forall\,x\geq0.
\]
On $1/2\le v\le2/3$ we have $b_v(x^2)\le0$ by \eqref{bvdef}, and differentiating \eqref{psidef} shows that
\[
 -\psi'(b_v(x^2))
 =\int_0^\infty4w^2e^{-w^2-2b_v(x^2)w}\dd w
 \ge\int_0^\infty4w^2e^{-w^2}\dd w=\sqrt\pi.
\]
Also, $\sin(\pi v)\ge\sqrt3/2$ on this interval and for any $x\geq0$,
\[
 \partial_vb_v(x^2)\stackrel{\eqref{bvdef}}{=}-\pi x q\sin(\pi v)
 \leq -\frac{\pi xq\sqrt{3}}{2},
 \qquad
 \int_{1/2}^{2/3}v(1-3v/2)\dd v=\frac5{432}.
\]
\begin{equation}\label{pvfv}
\partial_v F_v(x^2)
=e^{-\kappa x^2}\psi'(b_v(x^2))\partial_v b_v(x^2)
=e^{-\kappa x^2}\psi'(b_v(x^2))\partial_v b_v(x^2)
\geq e^{-\kappa x^2} \sqrt{\pi}\frac{\pi x q\sqrt{3}}{2}.
\end{equation}

The $J_0$ integrand is nonnegative on all of $[0,2/3]$; restricting it to
$[1/2,2/3]$ therefore proves
\begin{equation}\label{eq:Jstar-lower}
 J_*(x^2)\stackrel{\eqref{eq:star-kernels}}{\ge} J_0(x^2)
 \stackrel{\eqref{pvfv}}{\ge} c_0\frac{\sigma}{\sqrt{d}}\cdot x e^{-\kappa x^2},\,\forall\,x\geq0,\qquad
 c_0\colonequals\frac{5\pi\sqrt\pi\sqrt{3}}{4\cdot 432}.
\end{equation}
Moreover,
\begin{equation}\label{eq:Kstar-envelope}
 \mathsf K_*(x^2,y^2)
 \le\frac{1+4j_{(1)}/(3e)}{d}\exp\left[-\kappa(x^2+y^2)+\frac{\sigma}{d}xy\right],\qquad\forall\,x,y\geq0.
\end{equation}
Indeed, $K_0$ has the same upper bound with $1/d$ in front of the exponential:
\begin{flalign*}
\mathsf K_0(x^2, y^2)
&\stackrel{\eqref{eq:radial-K}\wedge\eqref{adef}}{=}\frac{1}{2d}e^{-(x^2 + y^2)\sigma^2 /d}
 \int_0^{2/3}(2+3v)e^{-2(\sigma/d)xy\cos(\pi v)}\dd v\\
 &\leq\frac{1}{2d}e^{-(x^2 + y^2)\sigma^2 /d}e^{(\sigma/d)xy}
 \int_0^{2/3}(2+3v)\dd v
 =\frac{1}{d}e^{-(x^2 + y^2)\sigma^2/d}e^{(\sigma/d)xy},\qquad\forall\,x,y\geq0.
\end{flalign*}
\[
 4\mathsf B(x^2,y^2)
 \stackrel{\eqref{bdef}}{=}\frac{4j_{(1)}\sigma}{d^2}xy
 e^{-(x^2+y^2)(\sigma^2 /d)-2\sigma xy/d}
 \leq 
  \frac{4j_{(1)}}{3de}
 e^{-(x^2+y^2)\sigma^2 /d+\sigma xy/d},
 \qquad\forall\,x,y\geq0,
\]
using $ae^{-(3/2)a}\leq 2/(3e)$ $\forall$ $a\geq0$, with $a=2\sigma xy/d$.  So, the $\mathsf K_0$ and $\mathsf B$ bounds prove \eqref{eq:Kstar-envelope} by \eqref{eq:star-kernels}.  

Equations
\eqref{eq:Jstar-lower}--\eqref{eq:Kstar-envelope}, completion of the
square, and Mills' inequality bound the complement of $[0,5]^2$.  To see this we estimate
 
\begin{equation}\label{ksineq}
\begin{aligned}
&\iint_{[0,\infty)^2\setminus[0,5]^2}e^{-x^2 - y^2}\frac{[K_*(x^2, y^2)]^2}{\widetilde J_*(x)\widetilde J_*(y)}\dd x \dd y\\
&\qquad\stackrel{\eqref{eq:Jstar-lower}\wedge\eqref{eq:Kstar-envelope}}{\leq}
\frac{d}{\sigma^2}c_0^{-2}\frac{(1+4j_{(1)}/(3e))^2}{d^2}\iint_{[0,\infty)^2\setminus[0,5]^2}
e^{-x^2 - y^2}
e^{-2\kappa(x^2 +y^2)+2\sigma xy/d}
e^{\kappa x^2}e^{\kappa y^2}\dd x\dd y\\
&\qquad= \frac{(1+4j_{(1)}/(3e))^2}{\sigma^2 d c_0^2}\iint_{[0,\infty)^2\setminus[0,5]^2}
e^{-(\kappa+1)(x^2 +y^2)+2\sigma xy/d}\dd x\dd y\\
&\qquad\stackrel{\eqref{eq:endpoint-parameters}}{=} \frac{(1+4j_{(1)}/(3e))^2}{\sigma^2 d c_0^2}\iint_{[0,\infty)^2\setminus[0,5]^2}
\exp\Big(\frac{-x^2 -y^2+2\sigma xy}{d}\Big)\dd x\dd y\\
&\qquad\leq
\frac{(1+4j_{(1)}/(3e))^2}{\sigma^2 d c_0^2}2\int_{5}^{\infty}\int_{0}^{\infty}\exp\Big(-\frac{x^2 + y^2}{1+\sigma}\Big)\dd x\dd y\\
&\qquad\leq
\frac{(1+4j_{(1)}/(3e))^2}{\sigma^2 d c_0^2}\sqrt{\pi(1+\sigma)}\frac{e^\frac{-5^2}{ 1+\sigma}(1+\sigma)}{2\cdot 5}
=\frac{(1+4j_{(1)}/(3e))^2}{\sigma^2 d c_0^2}\frac{\sqrt{\pi}(1+\sigma)^{3/2}e^{-\frac{5^2}{1+\sigma}}}{10}.
\end{aligned}
\end{equation}
The penultimate inequality used \eqref{eq:endpoint-parameters} and
\[
\frac{x^2 +y^2-2\sigma xy}{d}
=\frac{(x^2 +y^2)(1-\sigma)+\sigma(x-y)^2}{d}
\geq\frac{x^2 +y^2}{1+\sigma}
\]

Inequality \eqref{ksineq} is used at line 701 of the code.  By definition of $\mathsf E_*$ from \eqref{eq:Estar-def}, we have
\begin{flalign*}
|\mathsf E_*(x^2,y^2)|^2
&=|\mathsf K_*(x^2,y^2) - u_*(x^2)-u_*(y^2)+C_*|^2\\
&\leq2|\mathsf K_*(x^2,y^2)|^2 
+2(u_*(x^2)+u_*(y^2)-C_*)^2
\stackrel{\eqref{a0bd}}{\leq}
2|\mathsf K_*(x^2,y^2)|^2 +2(81/25)^2.
\end{flalign*}
We can then use this inequality to upper bound $4\iint_{[0,\infty)^2\setminus[0,5]^2}e^{-x^2 - y^2}\frac{[\mathsf E_*(x^2, y^2)]^2}{\widetilde J_*(x)\widetilde J_*(y)}\dd x \dd y$ in lines 714, 728 and 701 of the code, together with \begin{equation}\label{gauineq}
\iint_{[0,\infty)^2\setminus[0,5]^2}
\frac{e^{-x^2 - y^2}\dd x \dd y}{\widetilde{J}_*(x)\widetilde{J}_*(y)}
\leq\frac{d}{c_0^2\sigma^2}
2\int_{5}^{\infty}\int_{0}^{\infty}
e^{-(1-\kappa)(x^2 + y^2)}\dd x \dd y
\leq\frac{d\sqrt{\pi}e^{-(1-\kappa)\cdot 5^2}}{10c_0^2\sigma^2(1-\kappa)^{3/2}},
\end{equation}
recalling $1-\kappa=(1-2\sigma^2)/d>0$ by \eqref{eq:endpoint-parameters} since $0<\sigma\leq1/2$.  For an interval range of values of $\sigma$, the code estimates a lower bound $C_*\geq c_{\rm lo}$.  Then by definition of $\mathsf R_*$ from \eqref{eq:Rstar-def},
\begin{flalign*}
|\mathsf R_*(x^2,y^2)|^2
&=|\mathsf K_*(x^2,y^2) - u_*(x^2)u_*(y^2)/C_*|^2
\leq2|\mathsf K_*(x^2,y^2)|^2 
+2(u_*(x^2)u_*(y^2)/C_*)^2\\
&\stackrel{\eqref{a0bd}}{\leq}
2|\mathsf K_*(x^2,y^2)|^2 +2\cdot(27/25)^4 / (c_{\rm lo})^2.
\end{flalign*}

We can then use this inequality to upper bound $4\iint_{[0,\infty)^2\setminus[0,5]^2}e^{-x^2 - y^2}\frac{[\mathsf R_*(x^2, y^2)]^2}{\widetilde J_*(x)\widetilde J_*(y)}\dd x \dd y$ in lines 714, 731 and 701 of the code, together with \eqref{ksineq} and \eqref{gauineq}.

These estimates and \eqref{eq:cert-bounds} via \eqref{eq:hall-def}, \eqref{eq:hmz-def} are certified by the output of \verb!plur_cert.py!:

\begin{verbatim}
sigma interval                  [0.023255813953, 0.500000000000]
x cutoff / panels              5 / 1000
sigma panels                   640
uniform C_* enclosure          [0.683223606173, 0.994400009077]
upper H^2 finite worst          3.85644837674 on [0.4990318404,0.5000000000]
upper H^2 analytic tail         0.00724743483992
upper H^2 certified             3.86369581158
upper H certified               1.96562860469
lower H^2 finite worst          0.677401054168 on [0.0321787765,0.0325673567]
lower H^2 analytic tail         9.62724411884e-05
lower H^2 certified             0.677497326609
lower H certified               0.823102257687
runtime                         217.83 seconds
CERTIFIED: uniform upper mean-zero H<2 and unrestricted lower H<2.
\end{verbatim}

\end{proof}

\section{MAX-3-CUT Hardness}

\begin{proof}[Proof of Theorem \ref{thm:main}]
Let $\sigma\in(0,1/2]$.  The case $0<\sigma<1/43$ is covered by \cite{Heilman2023}.  So, it suffices to let $1/43\leq \sigma\leq 1/2$.  Apply the strengthened inequality \eqref{eq:endpoint-joint-arc} for each
pair $(t,s)$ of radial variables and integrate in $t,s$.  For
$(\Theta,-\Theta)$, all six radial arc lengths equal
$b=1/3$, corresponding to $3K_a(b,b)$.  
Recall that $x_i\colonequals p_i-b$, $y_i\colonequals q_i-b$, and $p_i,q_i\colon(0,\infty)\to[0,1]$ are defined after \eqref{eq:bilinear-target}, for all $i=1,2,3$, and below $\langle\cdot,\cdot\rangle$ denotes the $L_2$ inner product on $(0,\infty)$ with respect to the exponential measure $\nu$.
Circular rearrangement with \eqref{eq:integrated-star-form} gives
\begin{equation}\label{eq:integrated-lower}
 Q_\sigma^{(2)}(A,B)-Q_\sigma^{(2)}(\Theta,-\Theta)
 \ge\sum_{i=1}^3\Big(
 \ip{x_i}{\mathcal K_*y_i}
 +\ip{x_i}{M_{J_*}x_i}+\ip{y_i}{M_{J_*}y_i}
 \Big),
\end{equation}
with $\mathsf K_*$ defined in \eqref{eq:star-kernels}.  For each $1\le i\le 3$, set $z_i^+\colonequals(x_i+y_i)/\sqrt2$ and
$z_i^-\colonequals(x_i-y_i)/\sqrt2$.  By \eqref{eq:equal-radial-means},
$\int_{0}^{\infty} z_i^-(r)\dd\nu(r)=0$.  The self-adjointness of $\mathcal K_*$ by \eqref{eq:star-kernels} shows the
$i$-th summand in \eqref{eq:integrated-lower} is
\[
 \frac12\ip{z_i^+}{\mathcal K_*z_i^+}
 -\frac12\ip{z_i^-}{\mathcal K_*z_i^-}
 +\ip{z_i^+}{M_{J_*}z_i^+}+\ip{z_i^-}{M_{J_*}z_i^-}.
\]
Equations \eqref{eq:endpoint-lower-operator} and
\eqref{eq:endpoint-upper-operator}, together with
Lemma~\ref{lem:certificate}, bound this from below by
\begin{equation}\label{finaleq}
 \left(1-\frac{h_{\rm lower}}2\right)\ip{z_i^+}{M_{J_*}z_i^+}
 +\left(1-\frac{h_{\rm upper}}2\right)
   \ip{z_i^-}{M_{J_*}z_i^-}\ge0.
\end{equation}
This proves the bilinear inequality \eqref{eq:bilinear-target} since \eqref{eq:integrated-lower} is $\geq0$.  Theorem~\ref{thm:main} follows.
\end{proof}
\begin{remark}
The uniqueness part of Theorem \ref{thm1} follows by checking that, in the above proof, equality implies that $p=q=(1/3,1/3,1/3)$ a.e., and then the circular rearrangement argument implies: for almost every $r>0$, each set
$\Omega_i\cap rS^1$ is an arc of Haar measure $1/3$; and for each $1\leq i\leq 3$, the arc's midpoint does not depend on $r$, two distinct midpoints differ by an angle of $2\pi/3$.  Hence the partition agrees, up to
rotation and permutation of the cells, with the standard simplex
partition $\Theta$.
\end{remark}

\appendix
\section{Positive correlation Standard Simplex Inequality}
\label{sec:positive-appendix}

In this appendix we prove the three-set Standard Simplex Conjecture for
positive correlations in $(0,2/5]$.  Positive correlation is a
maximization problem, so the circular rearrangement uses aligned rather
than opposing arcs.  We retain the first angular mode separately, just as
in the negative-correlation proof, but its contribution now decreases the
cross coefficient and increases the diagonal penalty.

\begin{theorem}[positive-correlation Standard Simplex inequality]
\label{thm:positive-companion}
Let $n\geq2$, let $0<\rho\leq2/5$, and let
$\Omega=(\Omega_1,\Omega_2,\Omega_3)$ be a measurable partition of
$\R^n$ satisfying
$
 \gamma_n(\Omega_i)=\frac13
$
$\forall$ $1\leq i\leq 3$.  Then
\begin{equation}\label{eq:positive-main}
 Q_\rho^{(n)}(\Omega,\Omega)
 \leq Q_\rho^{(2)}(\Theta,\Theta).
\end{equation}
\end{theorem}

The interval $0<\rho<1/10$ is already contained in
\cite[Theorem~1.5]{Heilman2023}.  We therefore fix
\[
 \frac1{10}\leq\rho\leq\frac25,\qquad
 d\colonequals1-\rho^2,\qquad
 \kappa\colonequals\frac{\rho^2}{d}.
\]

\subsection{Sharpened joint arc inequality}

In this section we prove a positive correlation variant of Lemma \ref{lem:joint-arc}.

For any $p,q\in[0,1]$, let $E_p,E_q\subset\mathbb T$ be centered arcs of
normalized lengths $p,q$, and put
\begin{equation}\label{eq:positive-L-def}
 L_a(p,q)\colonequals
 \iint_{E_p\times E_q}
 \frac{e^{a\cos(\theta-\phi)}}{I_0(a)}
 \dd\mu(\theta)\dd\mu(\phi),\qquad a\geq0.
\end{equation}
The Fourier calculation leading to \eqref{eq:K-def} gives
\begin{align}
 L_a(p,q)
 &=pq+2\sum_{m\geq1}\frac{I_m(a)}{I_0(a)}
   \frac{\sin(\pi mp)\sin(\pi mq)}{\pi^2m^2},
 \label{eq:positive-L-Fourier}\\
 \partial_{pq}L_a(p,q)
 &=\frac{e^{a\cos(\pi(p-q))}
          +e^{a\cos(\pi(p+q))}}{2I_0(a)}.
 \label{eq:positive-L-mixed}
\end{align}
Recall that $b=1/3$ and set
\[
 k_+(a)\colonequals L_a(b,b),\qquad
 P_+(a)\colonequals\partial_1L_a(b,b),
\]
\begin{equation}\label{eq:positive-tau-w}
 \tau_+(a)\colonequals4P_+(a)-3k_+(a),\qquad
 w(a)\colonequals3k_+(a)-P_+(a).
\end{equation}

\begin{lemma}[aligned-arc bound]\label{lem:positive-aligned}
For every $a\geq0$ and $p,q\in\Delta_3$, put
$x=p-b\1$ and $y=q-b\1$.  Then
\begin{equation}\label{eq:positive-aligned}
 \sum_{i=1}^3L_a(p_i,q_i)-3L_a(b,b)
 \leq\tau_+(a)\ip{x}{y}
       -w(a)(\norm{x}^2+\norm{y}^2).
\end{equation}
Moreover, $w(a)>0$ when $a>0$.
\end{lemma}

\begin{proof}
When $a=0$, one has $L_0(p,q)=pq$, $\tau_+(0)=1$, and
$w(0)=0$, so \eqref{eq:positive-aligned} is an equality.  Henceforth
assume $a>0$.
Put
\[
 r_+(u)\colonequals\frac{e^{a\cos(\pi u)}}{I_0(a)}.
\]
This function is even, $2$-periodic, and decreasing on $[0,1]$.
Integrating \eqref{eq:positive-L-mixed}, exactly as in
\eqref{eq:kP-integrals}, gives
\[
 P_+=\frac12\int_0^{2/3}r_+(u)\dd u,\qquad
 k_+=\frac12\int_0^{2/3}(2/3-u)r_+(u)\dd u.
\]
Consequently,
\[
 P_+-3k_+
 =\frac12\int_0^{2/3}(3u-1)r_+(u)\dd u<0,
\]
which proves $w>0$.  The proof of Lemma~\ref{lem:joint-arc} now
applies with all monotonicity inequalities reversed.  We give the
details needed for the reversal.  Since
$\int_0^1r_+(u)\dd u=1$, the substitution
$v=2/3+3u^2/4$ gives
\[
\begin{aligned}
 \tau_+
 &=\int_0^{2/3}(1+3u/2)r_+(u)\dd u\geq1,\\
 \tau_+-2w
 &=\frac92\int_0^{2/3}u r_+(u)\dd u
 \geq r_+(2/3).
\end{aligned}
\]
Indeed, the first inequality is equivalent to
\[
 \frac32\int_0^{2/3}u r_+(u)\dd u
 \geq\int_{2/3}^1r_+(v)\dd v,
\]
and after the displayed substitution its left side minus its right side
is
\[
 \frac32\int_0^{2/3}u
 [r_+(u)-r_+(2/3+3u^2/4)]\dd u\geq0.
\]
The second inequality follows because $(9u/2)\dd u$ is a probability
measure on $[0,2/3]$ and $r_+$ is decreasing.  Now set
\[
 \Phi_+(p,q)
 \colonequals L_a(p,q)-\tau_+pq
   -w(p(1-p)+q(1-q)),\qquad\forall\,p,q\in[0,1].
\]
It is nonpositive on the boundary of $[0,1]^2$.  If
$G_+(\delta)\colonequals\int_0^\delta r_+(u)\dd u$, the first-derivative formulas
corresponding to \eqref{eq:K-first-p}--\eqref{eq:K-first-q} show that an
off-diagonal interior critical point satisfies
\[
 G_+(\delta)=(\tau_++2w)\delta.
\]
At such a point,
\[
 \partial_{(1,-1)}^2\Phi_+
 =-2\left(r_+(\delta)-\frac{G_+(\delta)}{\delta}\right)>0,
\]
because $r_+$ is decreasing; hence it cannot be a local maximum.
On the diagonal, $g_+(p)\colonequals\Phi_+(p,p)$ for all $p\in[0,1]$ satisfies
\[
 g_+(0)=g_+(b)=g_+'(b)=0,\qquad
 g_+(1)=1-\tau_+\leq0,
\]
and
\[
 g_+''(p)=2(r_+(2p)-\tau_++2w).
\]
The set on which $g_+''\leq0$ is an interval $[s,1-s]$, and the second
inequality above gives $s\leq b$.  Concavity and $g_+'(b)=0$ show that
$g_+\leq g_+(b)=0$ on this middle interval.  On each complementary
interval $g_+$ is convex, so it lies below the chord joining its
nonpositive endpoint values.  Hence $g_+\leq0$ on $[0,1]$.
Thus $\Phi_+\leq0$ on the square.  Summing this scalar
inequality over the three coordinates and using the identities at the
end of the proof of Lemma~\ref{lem:joint-arc} proves
\eqref{eq:positive-aligned}.
\end{proof}

We now provide a positive correlation version of Lemma \ref{lem:first-mode-gain}.  Write
\[
 H(p,q)\colonequals
 \frac{\sin(\pi p)\sin(\pi q)}{\pi^2}.
\]
Since $K^{(1)}(p,q)=pq-H(p,q)$ by \eqref{eq:first-mode-K}, the
first-mode inequality \eqref{eq:first-mode-gain} implies
\begin{equation}\label{eq:positive-first-mode}
 \sum_{i=1}^3(p_iq_i+H(p_i,q_i))
 -3(b^2+H(b,b))
\leq
 (2-\tau_{(1)}-2j_{(1)})\ip{x}{y}
 -\frac32j_{(1)}(\norm{x}^2+\norm{y}^2).
\end{equation}
Indeed, subtract \eqref{eq:first-mode-gain} from
$2\sum_{i=1}^3 p_iq_i-6b^2=2\langle x,y\rangle$.

Define $\theta(a)$ as in \eqref{eq:theta-def}.  The exact decomposition
\begin{equation}\label{eq:positive-angular-decomposition}
 \frac{e^{a\cos(\pi u)}}{I_0(a)}
 =
 \frac{e^{-a}}{I_0(a)}
 (1+a(1+\cos(\pi u)))+r_{+,{\rm rem}}(u)
\end{equation}
has
\[
 r_{+,{\rm rem}}(u)
 =\frac{e^{-a}}{I_0(a)}
 \left[e^{a(1+\cos(\pi u))}
       -1-a(1+\cos(\pi u))\right].
\]
The remainder is nonnegative and decreasing on $[0,1]$.  Apply
Lemma~\ref{lem:positive-aligned} homogeneously to the constant and
remainder terms in \eqref{eq:positive-angular-decomposition}, and apply
\eqref{eq:positive-first-mode} to $1+\cos(\pi u)$.  By linearity we get the following improvement to Lemma \ref{lem:positive-aligned}:

\begin{corollary}[sharpened aligned-arc bound]
\label{cor:positive-sharpened}
For every $a\geq0$ and $p,q\in\Delta_3$,
\begin{equation}\label{eq:positive-sharpened}
 \sum_{i=1}^3L_a(p_i,q_i)-3L_a(b,b)
 \leq(\tau_+(a)-4\theta(a))\ip{x}{y}
-(w(a)+\theta(a))
       (\norm{x}^2+\norm{y}^2).
\end{equation}
\end{corollary}

\subsection{Positive radial operator}

Define $R_0$ as in \eqref{eq:R0-def}.  Put
$a(t,s)=2\rho\sqrt{ts}/d$.  Define
\begin{align}
 \mathsf K_{0,+}(t,s)
 &\colonequals R_0(t,s)\tau_+(a(t,s))
 =\frac{e^{-\kappa(t+s)}}{2d}
   \int_0^{2/3}(2+3v)e^{a(t,s)\cos(\pi v)}\dd v,
 \label{eq:positive-K0}\\
 W_0(t)
 &\colonequals
 \int_0^\infty R_0(t,s)w(a(t,s))\dd\nu(s).
 \label{eq:positive-W0}
\end{align}
Use the same first-mode kernel $\mathsf B$ and its marginal
$B_{(1)}$ from \eqref{bdef}--\eqref{eq:B1-def}, with $\sigma$ replaced
by $\rho$, and set
\begin{equation}\label{eq:positive-star-kernels}
 \mathsf K_+(t,s)
 \colonequals\mathsf K_{0,+}(t,s)-4\mathsf B(t,s),
 \qquad
 J_+(t)\colonequals W_0(t)+B_{(1)}(t).
\end{equation}
Integrating \eqref{eq:positive-sharpened} in the two radial variables
therefore produces
\begin{equation}\label{eq:positive-radial-form}
 \sum_{i=1}^3
 \Big(\ip{x_i}{\mathcal K_+x_i}
       -2\ip{x_i}{M_{J_+}x_i}\Big),
\end{equation}
where $\mathcal K_+$ has kernel $\mathsf K_+$.

We now record the formulas used by the verification code.  Put
\[
 q=\frac{\rho x}{\sqrt d},\qquad
 \psi(z)=1-\sqrt\pi z e^{z^2}\operatorname{erfc}(z),
\]
\[
 S(q)\colonequals-\psi'(q)
 =\sqrt\pi(1+2q^2)e^{q^2}\operatorname{erfc}(q)-2q,
\]
\begin{equation}\label{eq:positive-H-def}
 H_+(q)\colonequals
 \frac{\pi}{2}\int_0^{2/3}
 v(1-3v/2)\sin(\pi v)
 [-\psi'(-q\cos(\pi v))]\dd v.
\end{equation}
Integration by parts in \eqref{eq:positive-W0}, followed by
\eqref{eq:B1-closed}, gives
\begin{equation}\label{eq:positive-J-formula}
 \widetilde J_+(x)
 \colonequals\frac{J_+(x^2)}x
 =
 e^{-\kappa x^2}\frac{\rho}{\sqrt d}
 \left[H_+(q)+\frac{j_{(1)}}2S(q)\right].
\end{equation}
Also, if
\[
 T_+(a)\colonequals
 \frac12\int_0^{2/3}(2+3v)e^{a\cos(\pi v)}\dd v,
\]
then
\begin{equation}\label{eq:positive-K-formula}
 \mathsf K_+(x^2,y^2)
 =
 \frac{e^{-\kappa(x^2+y^2)}}d
 \left[T_+\!\left(\frac{2\rho xy}{d}\right)
 -2j_{(1)}\frac{2\rho xy}{d}
  e^{-2\rho xy/d}\right].
\end{equation}

For a function $\ell_\rho$, additive centering does not change the
quadratic form on the mean-zero subspace.  We use
\begin{equation}\label{eq:positive-center}
 \ell_\rho(x)=\sum_{n=0}^4c_n(\rho)x^n,\qquad
 c_n(\rho)=\sum_{m=0}^4C_{nm}\rho^m,
\end{equation}
where $C$ is the binary64 matrix obtained by parsing the following
round-trip decimal strings:
\[
\resizebox{0.95\textwidth}{!}{$
\begin{pmatrix}
 .5001276609531384&-.04839192730072833& .3420752274381793&
 -.02840693452661011& .3929215251696178\\
 .0020711604332492072& .11278766504365427& .42921652164107776&
 -1.4633497642462785&2.392753763954644\\
 -.00657914712574693& .1533613436474771&-1.6052209990260096&
 4.917034579755538&-7.5025494501524435\\
 .00509284321106939&-.11817713264541818& .9939294369107601&
 -4.02653695935449&4.798761639440846\\
 -.00111590812093279& .02578645463758624&-.21531279149872498&
 .776584990984738&-.6977289068083498
\end{pmatrix}$}.
\]
These numbers merely specify an additive center; no approximation
property of them is assumed.  Set
\begin{equation}\label{eq:positive-HS}
 \mathfrak h_+^2(\rho)
 =
 4\int_0^\infty\int_0^\infty e^{-x^2-y^2}
 \frac{[\mathsf K_+(x^2,y^2)
 -\ell_\rho(x)-\ell_\rho(y)]^2}
 {\widetilde J_+(x)\widetilde J_+(y)}
 \dd x\dd y.
\end{equation}
The weighted Hilbert--Schmidt argument \eqref{eq:weighted-HS} gives,
for every bounded $z$ with $\int z\dd\nu=0$,
\begin{equation}\label{eq:positive-operator-bound}
 \ip{z}{\mathcal K_+z}
 \leq\mathfrak h_+(\rho)\ip{z}{M_{J_+}z}.
\end{equation}

\begin{lemma}[uniform directed positive certificate]
\label{lem:positive-certificate}
Uniformly for all $1/10\leq\rho\leq2/5$,
\[
 \mathfrak h_+^2(\rho)<3.9993,\qquad
 \mathfrak h_+(\rho)<1.99983<2.
\]
\end{lemma}

\begin{proof}
The program \path{plur_cert_positive.py} uses Arb balls for all
transcendental tables and outward-rounded binary64 interval arithmetic
for the remaining operations.  It uses $64$ second-order midpoint panels
in the angular variable, $4096$ nonuniform $q$-intervals, $16384$
nonuniform $a$-intervals, and $65536$ nonuniform exponential
intervals.  Here the nonuniform intervals $[q_i,q_{i+1}]$ use $q_i=q_{\rm max}(i/4096)^2$ for each $0\leq i\leq 4096$.  The intervals for $a$ and the exponential function $z\mapsto e^{-z}$ are constructed similarly.  More specifically, $
q_{\max}=\frac{(2/5)7}{\sqrt{1-(2/5)^2}},
$
$
a_{\max}=\frac{2(2/5)7^2}{1-(2/5)^2},
$
$
z_{\max}=49.
$
The auxiliary meshes are
$
q_i=q_{\max}(i/4096)^2,
$
$
a_i=a_{\max}(i/16384)^2,
$
$
z_i=z_{\max}(i/65536)^2.
$
The square $0\leq x,y\leq7$ is divided into $1600$ radial panels in each coordinate.
The correlation $\rho$ intervals begin with 30 uniform intervals of width $.01$ in $[.1,.4]$.  If a slab does not certify $\mathfrak h_+^2(\rho)<4$, it is bisected dyadically until the upper bound is below $4$.

For completeness, we give the analytic tails.  Integration by parts in
\eqref{eq:positive-H-def}, followed by restriction to
$0\leq v\leq1/2$, gives
\begin{equation}\label{eq:positive-J-lower}
 \widetilde J_+(x)\geq
 c_+\frac{\rho}{\sqrt d}e^{-\kappa x^2},\qquad
 c_+\colonequals\frac{6-\pi}{4\pi^{3/2}}.
\end{equation}
Indeed, $-\psi'(-q\cos(\pi v))\geq\sqrt\pi$ on this interval and
\[
 \int_0^{1/2}v(1-3v/2)\sin(\pi v)\dd v
 =\frac{6-\pi}{2\pi^3}.
\]
Moreover, $ae^{-a}\leq1/e$ and $\cos(\pi v)\leq1$ give
\begin{equation}\label{eq:positive-K-envelope}
 |\mathsf K_+(x^2,y^2)|
 \leq\frac{1+2j_{(1)}/e}{d}
 e^{-\kappa(x^2+y^2)+(2\rho/d)xy}.
\end{equation}
After division by the two factors in
\eqref{eq:positive-J-lower}, the squared-kernel tail has exponential
factor
\[
 \exp\left[-\frac{x^2+y^2-4\rho xy}{d}\right].
\]
Since $\rho<1/2$,
\[
 \frac{x^2+y^2-4\rho xy}{d}
 =\frac{(y-2\rho x)^2}{d}
  +\frac{1-4\rho^2}{d}x^2.
\]
Extending the inner half-line to $\R$, using the union bound, and then
Mills' inequality controls the complement of $[0,7]^2$.  The additive
tails follow from \eqref{eq:positive-center},
\[
 \ell_\rho(x)^2\leq5\sum_{n=0}^4c_n(\rho)^2x^{2n},
\]
and the elementary Gaussian moment recurrence.  The certification output is
\begin{verbatim}
rho interval                    [0.100000000000, 0.400000000000]
x cutoff / panels              7 / 1600
initial rho panels             30
accepted rho slabs             90
total slab attempts            150
H_+^2 finite worst              3.99929063236 on [0.3862500000,0.3875000000]
H_+^2 kernel tail               3.33733813468e-07
H_+^2 additive tail             3.03117401848e-12
H_+^2 certified                 3.99929096609
H_+ certified                   1.99982273367
runtime                         141.11 seconds
CERTIFIED: uniform positive mean-zero H_+<2.
\end{verbatim}
This proves the lemma.
\end{proof}

\begin{proof}[Proof of Theorem~\ref{thm:positive-companion}]
The case $0<\rho<1/10$ was proven in \cite{Heilman2023}, so it suffices to consider $1/10\leq\rho\leq2/5$.  The positive-correlation
dimension-reduction theorem of
Heilman--Tarter~\cite{HeilmanTarter2021}, together with the cylinder
embedding, reduces the supremum to the planar problem.  Recall that $x_i\colonequals p_i-b$  and $p_i\colon(0,\infty)\to[0,1]$ are defined after \eqref{eq:bilinear-target}, for all $i=1,2,3$, and below $\langle\cdot,\cdot\rangle$ denotes the $L_2$ inner product on $(0,\infty)$ with respect to the exponential measure $\nu$.  The assumed condition $\gamma_{2}(A_i)=1/3$ for all $i=1,2,3$ implies
\[
 \int_0^\infty x_i(t)\dd\nu(t)=0,\qquad i=1,2,3.
\]
Circular rearrangement bounds each conditional self-overlap from above
by aligned circular arcs.  Applying Corollary~\ref{cor:positive-sharpened} and
integrating in the two radial variables gives
\[
 Q_\rho^{(2)}(A,A)-Q_\rho^{(2)}(\Theta,\Theta)
 \leq
 \sum_{i=1}^3
 \Big(\ip{x_i}{\mathcal K_+x_i}
       -2\ip{x_i}{M_{J_+}x_i}\Big).
\]
Equation \eqref{eq:positive-operator-bound} and
Lemma~\ref{lem:positive-certificate} bound the right side by
\[
 (\mathfrak h_+(\rho)-2)
 \sum_{i=1}^3\ip{x_i}{M_{J_+}x_i}\leq0.
\]
This proves \eqref{eq:positive-main}.
\end{proof}

\section{MAX-2-LIN(3) Hardness}

\begin{proof}[Proof of Theorem \ref{thm:max-2lin-3-hardness}]
The PCP verifier of \cite[Theorem~12 and page 35]{khot07}, with correlation parameter
$\rho$, has completeness
\[
 c(\rho)=\rho+\frac{1-\rho}{3}=\frac{1+2\rho}{3}.
\]
Moreover, additive folding \cite[Definition 22]{khot07} makes each function arising in the
soundness analysis balanced.  The positive-correlation Plurality Is
Stablest Theorem \ref{thm2} and the invariance principle \cite{mossel10} therefore bound the
soundness by
\[
 \operatorname{Stab}_\rho(\mathrm{Plur}_3)+o(1),
\]
where the three-candidate plurality stability is
\[
\operatorname{Stab}_\rho(\mathrm{Plur}_3)
 =
 \frac13+\frac{3}{4\pi^2}
 \left(
   \arccos(-\rho)^2-\arccos(\rho/2)^2
 \right).
\]
Taking $\rho=2/5$ gives $c(2/5)=3/5$ and
\[
\operatorname{Stab}_{2/5}(\mathrm{Plur}_3)
 =
 0.489434116402967\ldots.
\]
Absorbing the invariance-principle and reduction errors into
$\varepsilon$ gives the asserted gap.  Dividing the soundness by the
completeness gives
\[
 \frac{\operatorname{Stab}_{2/5}(\mathrm{Plur}_3)}{3/5}
 =
 0.815723527338279\ldots,
\]
as claimed.
\end{proof}

\noindent\textbf{Acknowledgements.}  Thanks to Elchanan Mossel for introducing me to the Plurality is Stablest Problem at MSRI (now SLMath) in 2011, and for helpful conversations over the years.  Thanks to Assaf Naor for helpful conversations over the years.  Thanks to Yeongwoo Hwang, Joe Neeman, Ojas Parekh, Kevin Thompson, and John Wright for helpful discussions.

ChatGPT 5.6 assisted in the preparation of this manuscript, including producing the spectral certificates in Propositions \ref{prop:certificate} and \ref{prop:k4-reduction} and Lemmas \ref{lem:certificate} and \ref{lem:positive-certificate}.

\end{document}